\RequirePackage{fix-cm}
\RequirePackage{rotating}
\documentclass[smallextended]{svjour3}       
\smartqed  
\usepackage{graphicx}
\usepackage{amsmath}
\usepackage[ruled,nofillcomment]{algorithm2e}
\usepackage{epsfig}
\usepackage[sectionbib, round]{natbib}
\usepackage{url}
\usepackage[misc,geometry]{ifsym}
\usepackage{textgreek}
\journalname{Information Retrieval}

\renewcommand{\log}{\lg}
\newcommand{\CSA}
    {\ensuremath{\mathsf{CSA}}}
\newcommand{\ILCP}
    {\ensuremath{\mathsf{ILCP}}}
\newcommand{\LCP}
    {\ensuremath{\mathsf{LCP}}}
\newcommand{\LILCP}
    {\ensuremath{\mathsf{LILCP}}}
\newcommand{\lookup}[1]
    {\ensuremath{\mathsf{lookup}\!\left( {#1} \right)}}
\newcommand{\occ}
    {\ensuremath{\mathit{occ}}}
\newcommand{\Oh}[1]
    {\ensuremath{\mathcal{O}\!\left( {#1} \right)}}
\newcommand{\SA}
    {\ensuremath{\mathsf{SA}}}
\newcommand{\DA}
    {\ensuremath{\mathsf{DA}}}
\newcommand{\search}[1]
    {\ensuremath{\mathsf{search}\!\left( {#1} \right)}}
\newcommand{\VILCP}
    {\ensuremath{\mathsf{VILCP}}}

\newcommand{\set}[1]{\ensuremath{\{ #1 \}}}
\newcommand{\abs}[1]{\ensuremath{\lvert #1 \rvert}}
\newcommand{\avg}[1]{\ensuremath{\overline{ #1 }}}

\newcommand{\findq}{\ensuremath{\mathsf{find}}}
\newcommand{\locateq}{\ensuremath{\mathsf{locate}}}
\newcommand{\countq}{\ensuremath{\mathsf{count}}}

\newcommand{\rank}{\ensuremath{\mathsf{rank}}}
\newcommand{\select}{\ensuremath{\mathsf{select}}}

\newcommand{\rmq}{\ensuremath{\mathsf{rmq}}}
\newcommand{\tf}{\ensuremath{\mathsf{tf}}}
\newcommand{\df}{\ensuremath{\mathsf{df}}}
\newcommand{\tfidf}{\ensuremath{\textsf{tf-idf}}}

\newcommand{\Brute}{\textsf{Brute}}
\newcommand{\BruteL}{\textsf{Brute-L}}
\newcommand{\BruteD}{\textsf{Brute-D}}
\newcommand{\Sada}{\textsf{Sada}}
\newcommand{\SadaCL}{\textsf{Sada-L}} 
\newcommand{\SadaCD}{\textsf{Sada-D}} 
\newcommand{\SadaIL}{\textsf{ILCP-L}} 
\newcommand{\SadaID}{\textsf{ILCP-D}} 
\newcommand{\WT}{\textsf{WT}}
\newcommand{\Grammar}{\textsf{Grammar}}
\newcommand{\LZ}{\textsf{LZ}}
\newcommand{\PDL}{\textsf{PDL}}
\newcommand{\PDLBC}{\textsf{PDL-BC}}
\newcommand{\PDLRP}{\textsf{PDL-RP}}

\newcommand{\PDLtopk}[1]{\textsf{PDL--#1}}
\newcommand{\SURF}{\textsf{SURF}}

\newcommand{\SadaR}{\textsf{Sada-RR}} 
\newcommand{\SadaPG}{\textsf{Sada-P-G}} 
\newcommand{\SadaPR}{\textsf{Sada-P-RR}} 
\newcommand{\SadaRG}{\textsf{Sada-RR-G}} 
\newcommand{\SadaRR}{\textsf{Sada-RR-RR}} 
\newcommand{\SadaG}{\textsf{Sada-Gr}} 
\newcommand{\sadaR}{\textsf{Sada-RS}} 
\newcommand{\sadaRS}{\textsf{Sada-RS-S}} 
\newcommand{\sadaD}{\textsf{Sada-RD}} 
\newcommand{\sadaDS}{\textsf{Sada-RD-S}} 
\newcommand{\sadaS}{\textsf{Sada-S-S}} 
\newcommand{\sadaSS}{\textsf{Sada-S}} 
\newcommand{\wt}{\textsf{ILCP}} 

\newcommand{\Enwiki}{\textsf{Enwiki}}
\newcommand{\Page}{\textsf{Page}}
\newcommand{\Revision}{\textsf{Revision}}
\newcommand{\Influenza}{\textsf{Influenza}}
\newcommand{\Swissprot}{\textsf{Swissprot}}
\newcommand{\DNA}{\textsf{DNA}}
\newcommand{\Concat}{\textsf{Concat}}
\newcommand{\Version}{\textsf{Version}}
\newcommand{\Wiki}{\textsf{Wiki}}

\begin{document}

\title{Document Retrieval on Repetitive String Collections
\thanks{Preliminary partial versions of this paper appeared in Proc. CPM 2013,
Proc. ESA 2014, and Proc. DCC 2015. Part of this work was done while the first 
author was at the University of Helsinki and the third author was at Aalto 
University, Finland.}
}


\author{Travis Gagie       \and
	Aleksi Hartikainen \and
        Kalle Karhu        \and
	Juha K{\"a}rkk{\"a}inen \and
        Gonzalo Navarro    \and
        Simon J. Puglisi   \and
        Jouni Sir\'en      \and
}

\authorrunning{T. Gagie et al.} 

\institute{
   Travis Gagie \at
   CeBiB --- Center of Biotechnology and Bioengineering, \\
   School of Computer Science and Telecommunications, Diego Portales University, Chile\\
   \email{travis.gagie@gmail.com}
\and
   Aleksi Hartikainen \at
   Google Inc, USA\\
   \email{ahartik@gmail.com}
\and
   Kalle Karhu \at
   Research and Technology, Planmeca Oy, Finland\\
   \email{kalle.karhu@iki.fi}
\and
   Juha K{\"a}rkk{\"a}inen \at
   Helsinki Institute of Information Technology, Department of Computer Science, University of Helsinki, Finland\\
   \email{tpkarkka@cs.helsinki.fi}
\and
   Gonzalo Navarro \at
   CeBiB --- Center of Biotechnology and Bioengineering, \\
   Department of Computer Science, University of Chile, Chile\\
   \email{gnavarro@dcc.uchile.cl}
\and
   Simon J. Puglisi \at
   Helsinki Institute of Information Technology, Department of Computer Science, University of Helsinki, Finland\\
   \email{puglisi@cs.helsinki.fi}
\and
   Jouni Sir\'en (\Letter) \at
   Wellcome Trust Sanger Institute, UK\\
   \email{jouni.siren@iki.fi}
}

\date{Received: date / Accepted: date}

\maketitle

\begin{abstract}
Most of the fastest-growing string collections today are repetitive, that is,
most of the constituent documents are similar to many others. As these collections keep growing,
a key approach to handling them is to exploit their repetitiveness, which can reduce 
their space usage by orders of magnitude.
We study the problem of indexing repetitive string collections in order to perform 
efficient document retrieval operations on them. Document retrieval problems 
are routinely solved by search engines on large natural language collections,
but the techniques are less developed on generic string collections. The case 
of repetitive string collections is even less understood, and there are very few existing
solutions. We develop two novel ideas, {\em interleaved LCPs} and {\em
precomputed document lists}, that yield highly compressed indexes solving
the problem of document listing (find all the documents where a string
appears), top-$k$ document retrieval (find the $k$ documents where a string 
appears most often), and document counting (count the number of documents
where a string appears). We also show that a classical data structure 
supporting the latter query becomes highly compressible on repetitive data. 
Finally, we show how the tools we developed can be combined to solve ranked 
conjunctive and disjunctive multi-term queries under the simple tf-idf model 
of relevance. We thoroughly evaluate the resulting techniques in various 
real-life repetitiveness scenarios, and recommend the best choices for
each case.

\end{abstract}

\section{Introduction}

Document retrieval on natural language text collections is a routine activity 
in web and enterprise search engines. It is solved with variants of the 
inverted index \citep{BCC10,BYRN11}, an immensely successful technology that 
can by now be considered mature. The inverted index has well-known limitations,
however: the text must be easy to
parse into {\em terms} or {\em words}, and queries must be sets of words or
sequences of words ({\em phrases}). Those limitations are acceptable in most
cases when {\em natural language} text collections are indexed, and they enable
the use of an extremely simple index organization that is efficient and 
scalable, and that has been the key to the success of Web-scale information
retrieval. 

Those limitations, on the other hand, hamper the use of the inverted
index in other kinds of string collections where partitioning the text into
words and limiting queries to word sequences is inconvenient, difficult, or 
meaningless:
DNA and protein sequences, source code, music streams, and even some East Asian
languages. Document retrieval queries are of interest in those string 
collections, but the state of the art about alternatives to the inverted index
is much less developed \citep{HPSTV13,NavACMcs14}.

In this article we focus on {\em repetitive string collections}, where most of
the strings are very similar to many others.
These types of collections arise naturally in scenarios like versioned
document collections (such as Wikipedia\footnote{{\tt www.wikipedia.org}} or 
the Wayback Machine\footnote{From the Internet Archive,
{\tt www.archive.org/web/web.php}}),
versioned software repositories, periodical data publications in text form
(where very similar data is published over and over),
sequence databases with genomes of individuals of the same species (which
differ at relatively few positions), and so on.
Such collections are the fastest-growing ones today. For example, genome
sequencing data is expected to grow at least as fast as astronomical,
YouTube, or Twitter data by 2025, exceeding Moore's Law rate by a wide
margin \citep{Plos15}.
This growth brings new scientific opportunities but it also creates new 
computational problems.  

A key tool for handling this kind of growth is to exploit repetitiveness to obtain size
reductions of orders of magnitude. An appropriate Lempel-Ziv compressor%
\footnote{Such as p7zip, {\tt http://p7zip.sourceforge.net}} can successfully
capture such repetitiveness, and version control systems have offered direct 
access to any version since their beginnings, by means of storing the edits of a
version with respect to some other version that is stored in full \citep{Roc75}.
However, document retrieval requires much more than retrieving individual
documents. In this article we focus on three basic document retrieval
problems on string collections:

\begin{description}
\item[{\em Document Listing:}] Given a string $P$, list the identifiers of all the $\df$ 
documents where $P$ appears.
\item[{\em Top-$k$ Retrieval:}] Given a string $P$ and $k$, list $k$ documents
where $P$ appears most often.
\item[{\em Document Counting:}] Given a string $P$, return the number $\df$ of
documents where $P$ appears.
\end{description}

Apart from the obvious case of information retrieval on East Asian and other
languages where separating words is difficult, these queries are relevant in
many other applications where string collections are maintained.
For example, in pan-genomics \citep{Marschall043430} we index the genomes of all the strains of an organism.  The index can be either a specialized data structure, such as a colored de Bruijn graph, or a text index over the concatenation of the individual genomes.  The parts of the genome common to all strains are called core; the parts common to several strains are called peripheral; and the parts in only one strain are called unique.  Given a set of DNA reads from an unidentified strain, we may want to identify it (if it is known) or find the closest strain in our database (if it is not), by identifying reads from unique or peripheral genomes (i.e., those that occur rarely) and listing the corresponding strains. This boils down to document listing and counting problems. In turn, top-$k$
retrieval is at the core of information retrieval systems, since the {\em term
frequency} $\tf$ (i.e., the number of times a pattern appears in a document) is a basic criterion to establish the relevance of a document for a query
\citep{BCC10,BYRN11}. On multi-term queries, it is usually combined with the
document frequency, $\df$, to compute \tfidf, a simple and popular
relevance model. Document counting is also important for data mining 
applications on strings (or {\em string mining} \citep{DPT12}), where the value
$\df/d$ of a given pattern, $d$ being the total number of documents, is its 
{\em support} in the collection. Finally, we will show that the best choice of 
document listing and top-$k$ retrieval algorithms in practice strongly depends 
on the $\df/\occ$ ratio, where $\occ$ is the number of times the pattern appears
in the collection, and thus the ability to compute $\df$ quickly allows for 
the efficient selection of an appropriate listing or top-$k$ algorithm at query
time. \cite{NavACMcs14} lists several other applications of these queries.

In the case of natural language, there exist various proposals to
reduce the inverted index size by exploiting the text
repetitiveness \citep{AF92,BEFHLMQS06,HYS09,HZS10,He:Sigir12,CFMNis16}.
For general string collections, the situation is much worse. Most of the
indexing structures designed for repetitive string collections
\citep{MNSV09,CFMPN10,CN11,CN12,KNtcs12,GGKNP12,GGKNP14,DJSS14,BCGPR15} support
only {\em pattern matching}, that is, they count or list the $\occ$ 
occurrences of a pattern $P$ in the whole collection. Of course one
can retrieve the $\occ$ occurrences and then answer any of our three document
retrieval queries, but the time will be $\Omega(\occ)$. Instead, there are 
optimal-time indexes for string collections that solve document listing in time
$\Oh{\abs{P}+\df}$ \citep{Mut02}, top-$k$ retrieval in time $\Oh{\abs{P}+k}$ 
\citep{NN12}, and document counting in time $\Oh{\abs{P}}$ \citep{Sad07}.
The first two solutions, however, use a lot of space even for classical, 
non-repetitive collections. While more compact representations have been
studied \citep{HPSTV13,NavACMcs14}, none of those is tailored to the repetitive 
scenario, except for a grammar-based index that solves document listing 
\citep{CM13}.

In this article we develop several novel solutions for the three document
retrieval queries of interest, tailored to repetitive string collections.
Our first idea, called {\em interleaved LCPs (ILCP)} stores the longest
common prefix (LCP) array of the documents, interleaved in the order of the
global LCP array. The ILCP turns out to have a number of
interesting properties that make it compressible on repetitive collections,
and useful for document listing and counting. Our second idea, {\em 
precomputed document lists (PDL)}, samples some nodes in the global suffix 
tree of the collection and stores precomputed answers on those. It then
applies grammar compression on the stored answers, which is effective when the
collection is repetitive. PDL yields very efficient solutions for document
listing and top-$k$ retrieval. Third, we show that a solution for document
counting \citep{Sad07} that uses just two bits per symbol (bps) in the worst 
case (which is unacceptably high in the repetitive scenario) turns out to be 
highly compressible when the collection is repetitive, and becomes the most
attractive solution for document counting. 
Finally, we show how the different components of our solutions can be
assembled to offer \tfidf\ ranked conjunctive and disjunctive multi-term
queries on repetitive string collections. 

We implement and experimentally compare several
variants of our solutions with the state of the art, including the solution
for repetitive string collections \citep{CM13} and some relevant solutions
for general string collections \citep{FN13,GNalenex15}. We consider various
kinds of real-life repetitiveness scenarios, and show which solutions are
the best depending on the kind and amount of repetitiveness, and the space 
reduction that 
can be achieved. For example, on very repetitive collections of up to
1~GB we perform 
document listing and top-$k$ retrieval in 10--100 microseconds per result
and using 1--2 bits per symbol. For counting, we use as little as
0.1 bits per symbol and answer queries in less than a microsecond.
Multi-term top-$k$ queries can be solved with a throughput of 100-200 
queries per second, which we show to be similar to that of
a state-of-the-art inverted index. Of course, we do not aim to compete with
inverted indexes in the scenarios where they can be applied (mainly, in
natural language text collections), but to offer similar functionality in
the case of generic string collections, where inverted indexes cannot be used.

This article collects our earlier results appearing in {\em CPM 2013}
\citep{GKNPS13}, {\em ESA 2014} \citep{NPS14}, and {\em DCC 2015}
\citep{GHKNPS15}, where we focused on exploiting repetitiveness in different
ways to handle different document retrieval problems. Here we present them 
in a unified form, considering the application of two new techniques (ILCP
and PDL) and an existing one \citep{Sad07} to the three problems (document listing,
top-$k$ retrieval, and document counting), and showing how they interact (e.g.,
the need to use fast document counting to choose the best document listing 
method). In this article we also consider a more complex document retrieval
problem we had not addressed before: top-$k$ retrieval of multi-word queries. 
We present an algorithm that uses our (single-term) top-$k$ retrieval and 
document counting structures to solve ranked multi-term conjunctive and 
disjunctive queries under the \tfidf\ relevance model.

The article is organized as follows (see Table~\ref{tab:org}). 
In Section~\ref{sec:prelim} we introduce
the concepts needed to follow the presentation. In Section~\ref{sec:ilcp} we
introduce the Interleaved LCP (ILCP) structure and show 
how it can be used for document 
listing and, with a different representation, for document counting. In
Section~\ref{sec:pdl} we introduce our second structure, Precomputed Document
Lists (PDL), and describe how it can be used for document listing and, with
some reordering of the lists, for top-$k$ retrieval. Section~\ref{sec:count}
then returns to the problem of document counting, not to propose a new data
structure but to study a known one \citep{Sad07}, which is found to be 
compressible in a
repetitiveness scenario (and, curiously, on totally random texts as well).
Section~\ref{sec:tfidf} shows how our developments can be combined to build
a document retrieval index that handles multi-term queries.
Section~\ref{sec:exp} empirically studies the performance of our solutions
on the three document retrieval problems, also comparing them with the
state of the art for generic string collections, repetitive or not, and
giving recommendations on which structure to use in each case. Finally,
Section~\ref{sec:concl} concludes and gives some future work directions.

\begin{table}[b]
\begin{center}
\begin{tabular}{l|c|c|c}
Problem & ILCP                         & PDL                        & Sadakane\\
\hline
Listing & Section~\ref{sec:repet}      & Section~\ref{sec:pdl-list} & \\
Top-$k$ &                              & Section~\ref{sec:topk}     & \\
Counting& Section~\ref{sec:ilcp-count} &                            & Section~\ref{sec:count}
\end{tabular}
\end{center}
\caption{The techniques we study and the document retrieval problems we
solve with them.}
\label{tab:org}
\end{table}

\section{Preliminaries} \label{sec:prelim}

\subsection{Suffix Trees and Arrays}

A large number of solutions for pattern matching or document retrieval on
string collections rely on the suffix tree \citep{Wei73} or the suffix array
\citep{MM93}. Assume that we have a collection of $d$ strings, each terminated with
a special symbol ``\$'' (which we consider to be lexicographically smaller than any other symbol), and let $T[1..n]$ be their concatenation. The
suffix tree of $T$ is a compacted digital tree where all the suffixes $T[i..n]$ are
inserted. Collecting the leaves of the suffix tree yields the suffix array, 
$\SA[1..n]$, which is an array of pointers to all the suffixes sorted in 
increasing lexicographic order, that is, $T[\SA[i]..n] < T[\SA[i+1]..n]$ for 
all $1 \le i < n$. To find all the $\occ$ occurrences of a string $P[1..m]$ in 
the collection, we traverse the suffix tree following the symbols of $P$ and
output the leaves of the node we arrive at, called the {\em locus} of $P$,
in time $\Oh{m+\occ}$. On a suffix
array, we obtain the range $\SA[\ell..r]$ of the leaves (i.e., of the suffixes
prefixed by $P$) by binary search, and then list the contents of the range,
in total time $\Oh{m\log n+\occ}$.

We will make use of compressed suffix arrays \citep{NM07}, which we will call
generically $\CSA$s. Their size in bits is denoted $\abs{\CSA}$, their time to
find $\ell$ and $r$ is denoted $\search{m}$, and their time to access any
cell $\SA[i]$ is denoted $\lookup{n}$. 
A particular version of the $\CSA$ that is tailored
for repetitive collections is the Run-Length Compressed Suffix Array (RLCSA) \citep{MNSV09}.

\subsection{Rank and Select on Sequences} \label{sec:ranksel}

Let $S[1..n]$ be a sequence over an alphabet $[1..\sigma]$. When $\sigma=2$ we
use $0$ and $1$ as the two symbols, and the sequence is called a bitvector.
Two operations of interest on $S$ are $\rank_c(S,i)$, which counts the number
of occurrences of symbol $c$ in $S[1..i]$, and $\select_c(S,j)$, which gives
the position of the $j$th occurrence of symbol $c$ in $S$. For bitvectors, one
can compute both functions in $\Oh{1}$ time using $o(n)$ bits on top of $S$
\citep{Cla96}. If $S$ contains $m$ $1$s, we can also represent it using
$m\lg\frac{n}{m}+\Oh{m}$ bits, so that $\rank$ takes $\Oh{\lg\frac{n}{m}}$ time
and $\select$ takes $\Oh{1}$ \citep{OS07}\footnote{This is achieved by using a 
constant-time $\rank$/$\select$ solution \citep{Cla96} to represent their 
internal bitvector $H$.}.

The wavelet tree \citep{GGV03} is a tool for extending bitvector representations
to sequences. It is a binary tree where the alphabet $[1..\sigma]$ is
recursively partitioned. The root represents $S$ and stores a
bitvector $W[1..n]$ where $W[i]=0$ iff symbol $S[i]$ belongs to the left child.
Left and right children represent a subsequence of $S$ formed by the symbols
of $[1..\sigma]$ they handle, so they recursively store a bitvector and so on
until reaching the leaves, which represent a single symbol. By giving
constant-time $\rank$ and $\select$ capabilities to the bitvectors associated
with the nodes, the wavelet tree can compute any $S[i]=c$, $\rank_c(S,i)$,
or $\select_c(S,j)$ in time proportional to the depth of the leaf of $c$.
If the bitvectors are represented in a certain compressed form \citep{RRR07},
then the total space is at most $n\lg\sigma + o(nh)$, where $h$ is the wavelet
tree height, independent of the way the alphabet is partitioned \citep{GGV03}.

\subsection{Document Listing} \label{sec:muthu}

Let us now describe the optimal-time algorithm of \cite{Mut02} for document 
listing. \citeauthor{Mut02} stores the suffix tree of $T$; a
so-called \emph{document array} \(\DA[1..n]\) of $T$, in which each cell \(\DA[i]\)
stores the identifier of the document containing \(T [\SA [i]]\); an array \(C
[1..n]\), in which each cell \(C [i]\) stores the largest value \(h < i\) such
that \(\DA[h] = \DA[i]\), or 0 if there is no such value $h$; and a data
structure supporting range-minimum queries (RMQs) over $C$,
$\rmq_C(i,j) = \arg \min_{i \le k \le j} C[k]$.  These data
structures take a total of $\Oh{n \log n}$ bits.  Given a pattern \(P
[1..m]\), the suffix tree is used to find the interval \(\SA [\ell..r]\) that
contains the starting positions of the suffixes prefixed by $P$.  It follows
that every value $C[i] < \ell$ in $C[\ell..r]$ corresponds to a
distinct document in $\DA[i]$. Thus a recursive algorithm finding all those
positions $i$ starts with $k = \rmq_C(\ell,r)$. If $C[k] \ge \ell$ it
stops. Otherwise it reports document $\DA[k]$ and continues recursively with
the ranges $C[\ell..k-1]$ and $C[k+1..r]$ (the condition $C[k] \ge \ell$ always
uses the original $\ell$ value). In total, the algorithm uses $\Oh{m + \df}$ time, where $\df$ is the number of documents returned.

\cite{Sad07} proposed a space-efficient version of this algorithm, using just
$|\CSA|+\Oh{n}$ bits. The suffix tree is replaced with a $\CSA$. The array 
$\DA$ is replaced with a bitvector $B[1..n]$
such that $B[i]=1$ iff $i$ is the first symbol of a document in $T$. Therefore
$\DA[i] = \rank_1(B,\SA[i])$ can be computed in constant time \citep{Cla96}. The
RMQ data structure is replaced with a variant \citep{FH11} that uses just 
$2n+o(n)$ bits and answers queries in constant time without accessing $C$.
Finally, the comparisons $C[k] \ge \ell$ are replaced by marking the documents
already reported in a bitvector $V[1..d]$ (initially all 0s), so that $V[\DA[i]]=1$ iff document 
$\DA[i]$ has already been reported. If $V[\DA[i]]=1$ the recursion stops, otherwise
it sets $V[\DA[i]]$, reports $\DA[i]$, and continues. This is correct as long as
the RMQ structure returns the leftmost minimum in the range, and the range 
$[\ell..k-1]$ is processed before the range $[k+1..r]$ \citep{NavACMcs14}. The
total time is then $\Oh{\search{m}+\df \cdot \lookup{n}}$.

\section{Interleaved LCP}
\label{sec:ilcp}

We introduce our first structure, the Interleaved LCP
(ILCP). The main idea is to interleave the longest-common-prefix (LCP) 
arrays of the documents, in the order given by the global LCP of the collection.
This yields long runs of equal values on repetitive collections, making the
ILCP structure run-length compressible. Then, we show that the classical 
document listing technique of \citet{Mut02}, designed to work on a completely different 
array, works almost verbatim over the ILCP array, and this yields a new 
document listing technique of independent interest for string 
collections. Finally, we show that a particular representation of the ILCP
array allows us to count the number of documents where a string appears without
having to list them one by one.

\subsection{The ILCP Array}

The longest-common-prefix array \(\LCP_S [1..\abs{S}]\) of a string $S$ is defined such that \(\LCP_S [1] = 0\) and, for \(2 \leq i \leq \abs{S}\), \(\LCP_S [i]\) is the length of the longest common prefix of the lexicographically \((i - 1)\)th and $i$th suffixes of $S$, that is, of $S[\SA_S[i-1]..\abs{S}]$ and $S[\SA_S[i]..\abs{S}]$, where $\SA_S$ is the suffix array of $S$.  We define the interleaved LCP array of $T$, \ILCP, to be the interleaving of the LCP arrays of the individual documents according to the document array.

\begin{definition}
Let $T[1..n] = S_1 \cdot S_2 \dotsm S_d$ be the concatenation of documents $S_j$,
$\DA$ the document array of $T$, and
$\LCP_{S_j}$ the longest-common-prefix array of string $S_j$. Then the
{\em interleaved LCP array of $T$} is defined, for all $1 \le i \le n$, as
\[\ILCP[i] ~~=~~ \LCP_{S_{\DA[i]}}\left[ \rank_{\DA[i]}(\DA,i) \right].\]
\end{definition}

That is, if the suffix $\SA[i]$ belongs to document $S_j$ (i.e., $\DA[i]=j$), 
and this is the $r$th suffix of $\SA$ that belongs to $S_j$ (i.e.,
$r = \rank_j(\DA,i)$), then $\ILCP[i] = \LCP_{S_j}[r]$. Therefore the order
of the individual $\LCP$ arrays is preserved in $\ILCP$.

\paragraph{Example}
Consider the documents $S_1=\texttt{"TATA\$"}$, $S_2=\texttt{"LATA\$"}$, and
$S_3=\texttt{"AAAA\$"}$. Their concatenation is $T=\texttt{"TATA\$LATA\$AAAA\$"}$, its suffix array is $\SA=\langle 15,10,5,14,9,4,13,12,11,7,2,6,8,3,1\rangle$
and its document array is $\DA=\langle
\mathit{3},\mathbf{2},1,\mathit{3},\mathbf{2},1,\mathit{3},\mathit{3},\mathit{3},\mathbf{2},1,\mathbf{2},\mathbf{2},1,1\rangle$.
The LCP arrays of the documents are $\LCP_{S_1}=\langle 0,0,1,0,2\rangle$,
$\LCP_{S_2}=\langle
\mathbf{0},\mathbf{0},\mathbf{1},\mathbf{0},\mathbf{0}\rangle$, and
$\LCP_{S_3}=\langle
\mathit{0},\mathit{0},\mathit{1},\mathit{2},\mathit{3}\rangle$. Therefore,
$\ILCP = \langle
\mathit{0},\mathbf{0},0,\mathit{0},\mathbf{0},0,\mathit{1},\mathit{2},\mathit{3},\mathbf{1},1,\mathbf{0},\mathbf{0},0,2\rangle$
interleaves the $\LCP$ arrays in the order given by $\DA$ (notice the fonts
above).

\medskip

The following property of $\ILCP$ makes it suitable for document retrieval.

\begin{lemma} \label{lem:leftmost}
Let $T[1..n] = S_1 \cdot S_2 \dotsm S_d$ be the concatenation of documents $S_j$,
$\SA$ its suffix array and $\DA$ its document array.
Let \(\SA [\ell..r]\) be the interval that contains the starting positions of suffixes prefixed by a pattern \(P [1..m]\). Then the leftmost occurrences of the distinct document identifiers in $\DA[\ell..r]$ are in the same positions as the values strictly less than $m$ in $\ILCP[\ell..r]$.
\end{lemma}

\begin{proof}
Let $\SA_{S_j}[\ell_j..r_j]$ be the interval of all the suffixes of $S_j$
starting with $P[1..m]$. Then $\LCP_{S_j}[\ell_j] < m$, as
otherwise $S_j[\SA[\ell_j-1]..\SA[\ell_j-1]+m-1] = S_j[\SA[\ell_j]..\SA[\ell_j]+m-1] = P$ as well,
contradicting the definition of $\ell_j$. For the same reason, it holds that
$\LCP_{S_j}[\ell_j+k] \ge m$ for all $1 \le k \le r_j-\ell_j$.

Now let $S_j$ start at position $p_j+1$ in $T$, where $p_j = \abs{S_1 \dotsm S_{j-1}}$.
Because each $S_j$ is terminated by ``\$'', the lexicographic
ordering between the suffixes $S_j[k..]$ in $\SA_{S_j}$ is the same as that of the
corresponding suffixes $T[p_j+k..]$ in $\SA$. Hence
$\langle \SA[i] \mid \DA[i] = j, 1\le i\le n \rangle =
\langle p_j+\SA_{S_j}[i] \mid 1\le i\le \abs{S_j} \rangle$. Or, put
another way, $\SA[i] = p_j + \SA_{S_j}[\rank_j(\DA,i)]$ whenever $\DA[i]=j$.

Now let $f_j$ be the leftmost occurrence of $j$ in $\DA[\ell..r]$. This
means that $\SA[f_j]$ is the lexicographically first suffix of $S_j$ that
starts with $P$. By the definition of $\ell_j$, it holds that
$\ell_j = \rank_j(\DA,f_j)$. Thus, by definition of $\ILCP$, it holds that
$\ILCP[f_j] = \LCP_{S_j}[\rank_j(\DA,f_j)] = \LCP_{S_j}[\ell_j] < m$, whereas
all the other $\ILCP[k]$ values, for $\ell \le k \le r$, where $\DA[k]=j$,
must be $\ge m$.
\qed
\end{proof}

\paragraph{Example}
In the example above, if we search for $P[1..2]=\texttt{"TA"}$, the resulting 
range is $\SA[13..15]=\langle 8,3,1 \rangle$. The corresponding range
$\DA[13..15]=\langle 2,1,1\rangle$ indicates that the occurrence at 
$\SA[13]$ is in $S_2$ and those in $\SA[14..15]$ are in $S_1$. According to
the lemma, it is sufficient to report the documents $\DA[13]=2$ and $\DA[14]=1$,
as those are the positions in $\ILCP[13..15] = \langle 0,0,2 \rangle$ with
values less than $\abs{P} = 2$.

\medskip

Therefore, for the purposes of document listing, we can replace the $C$ array
by \ILCP\ in Muthukrishnan's algorithm (Section~\ref{sec:muthu}): 
instead of recursing until we have listed
all the positions $k$ such that $C[k] < \ell$, we recurse until we list all
the positions $k$ such that $\ILCP[k] < m$. Instead of using it directly,
however, we will design a variant that exploits repetitiveness in the string
collection.

\subsection{ILCP on Repetitive Collections}\label{sec:runs}

The array $\ILCP$ has yet another property, which makes it
attractive for repetitive collections: it contains long runs of equal values.
We give an analytic proof of this fact under a model where 
a base document $S$ is generated at random under the very general A2 
probabilistic model of \cite{Szp93}\footnote{This model states that the 
statistical dependence of a symbol from previous ones tends to zero as 
the distance towards them tends to infinity. 
The A2 model includes, in particular, the Bernoulli 
model (where each symbol is generated independently of the context), stationary
Markov chains (where the probability of each symbol depends on the previous 
one), and $k$th order models (where each symbol depends on the $k$ previous 
ones, for a fixed $k$).}, and the 
collection is formed by performing some edits on $d$ copies of $S$.

\begin{lemma} \label{lem:rho}
Let $S[1..r]$ be a string generated under Szpankowski's A2 model.
Let $T$ be formed by concatenating $d$ copies
of $S$, each terminated with the special symbol ``\$'', and then carrying
out $s$ edits (symbol insertions, deletions, or substitutions) at
arbitrary positions in $T$ (excluding the `\$'s).
Then, almost surely (a.s.\footnote{This
is a very strong kind of convergence. A sequence $X_n$ tends to a value $\beta$
almost surely if, for every $\epsilon>0$, the probability that
$\abs{X_N/\beta-1} > \epsilon$ for some $N>n$ tends to zero as $n$ tends to infinity,
$\lim_{n\rightarrow \infty}\sup_{N>n} \mathrm{Pr}(\abs{X_N/\beta-1}>\epsilon)=0$.}),
the $\ILCP$ array of $T$ is formed by
$\rho \le r + \Oh{s\log (r+s)}$ runs of equal values.
\end{lemma}
\begin{proof}
Before applying the edit operations, we have $T = S_1 \dotsm S_d$
and $S_j=S\$$ for all $j$. At this point, $\ILCP$ is formed by at most $r+1$ runs of equal
values, since the $d$ equal suffixes $S_j[\SA_{S_j}[i]..r+1]$ must be contiguous
in the suffix array $\SA$ of $T$, in the area $\SA[(i-1)d+1..id]$.
Since the values $l=\LCP_{S_j}[i]$ are also equal, and $\ILCP$ values are
the $\LCP_{S_j}$ values listed in the order of $\SA$, it follows that
$\ILCP[(i-1)d+1..id]=l$ forms
a run, and thus there are $r+1=n/d$ runs in $\ILCP$.
Now, if we carry out $s$ edit operations on $T$, any $S_j$ will be of length
at most $r+s+1$.
Consider an arbitrary edit operation at $T[k]$. It changes all the suffixes
$T[k-h..n]$ for all $0 \le h < k$. However, since a.s. the string depth of a
leaf in the suffix tree of $S$ is $\Oh{\log (r+s)}$ \citep{Szp93}, the suffix
will possibly be moved in $\SA$ only for $h = \Oh{\log (r+s)}$. Thus, a.s., only
$\Oh{\log (r+s)}$ suffixes are moved in $\SA$, and possibly the corresponding
runs in $\ILCP$ are broken.
Hence $\rho \le r+\Oh{s\log (r+s)}$ a.s.
\qed
\end{proof}

Therefore, the number of runs depends linearly on the size of the base document
and the number of edits, not on the total collection size.
The proof generalizes the arguments of \cite{MNSV09}, which hold for uniformly
distributed strings $S$. There is also experimental evidence \citep{MNSV09} 
that, in real-life text collections, a small change to a string usually causes 
only a small change to its $\LCP$ array. Next we design a document listing 
data structure whose size is bounded in terms of $\rho$.

\subsection{Document Listing}\label{sec:repet}

Let \(\LILCP [1..\rho]\) be the array containing the partial sums of the
lengths of the $\rho$ runs in \ILCP, and let \(\VILCP [1..\rho]\) be the array containing the values in those runs. We can store \LILCP\ as a bitvector $L[1..n]$ with $\rho$ 1s, so that $\LILCP[i]=\select(L,i)$. Then $L$ can be stored 
using the structure of \cite{OS07} that requires $\rho\log(n/\rho)+\Oh{\rho}$ 
bits.

With this representation, it holds that $\ILCP[i] = \VILCP[\rank_1(L,i)]$. 
We can map from any position $i$ to its run $i'=\rank_1(L,i)$ in time 
$\Oh{\log(n/\rho)}$, and from any run $i'$ to its starting position in 
$\ILCP$, $i=\select(L,i')$, in constant time.

\paragraph{Example.} Consider the array
$\ILCP[1..15] = \langle 0,0,0,0,0,0,1,2,3,1,1,0,0,0,2\rangle$ of our running 
example. It has $\rho=7$ runs, so we represent it with 
$\VILCP[1..7]=\langle 0,1,2,3,1,0,2\rangle$ and
$L[1..15]= 100000111101001$.

\medskip

This is sufficient to emulate the document listing algorithm of \cite{Sad07} 
(Section~\ref{sec:muthu}) on a repetitive collection. We will use RLCSA as the 
$\CSA$.
The sparse bitvector $B[1..n]$ marking the document beginnings in $T$ will be
represented in the same way as $L$, so that it requires $d\log(n/d)+\Oh{d}$ bits and
lets us compute any value $\DA[i]=\rank_1(B,\SA[i])$ in time
$\Oh{\lookup{n}}$. Finally, we build the compact RMQ data structure 
\citep{FH11} on $\VILCP$, requiring $2\rho+o(\rho)$ bits. We note that this RMQ
structure does not need access to $\VILCP$ to answer queries.

Assume that we have already found the range $\SA[\ell..r]$ in
$\Oh{\search{m}}$ time. We compute \(\ell' = \rank_1 (L,\ell)\) and \(r' =
\rank_1 (L,r)\), which are the endpoints of the interval \(\VILCP [\ell'..r']\) containing the values in the runs in \(\ILCP [\ell..r]\).
Now we run \citeauthor{Sad07}'s algorithm on $\VILCP[\ell'..r']$. Each time we
find a minimum at $\VILCP[i']$, we remap it to the run $\ILCP[i..j]$, where
$i=\max(\ell,\select(L,i'))$ and $j=\min(r,\select(L,i'+1)-1)$.
For each $i \le k \le j$, we compute $\DA[k]$ using $B$ and RLCSA as
explained, mark it in $V[\DA[k]] \leftarrow 1$, and report it. If, however,
it already holds that $V[\DA[k]]=1$, we stop the recursion.
Figure~\ref{fig:ilcp} gives the pseudocode.

\begin{figure}[t]\centering
\begin{minipage}[t]{0.5\linewidth}
\begin{tabbing}
mm\=mm\=mm\=mm\= \kill
\textbf{function} $\operatorname{listDocuments}(\ell, r)$ \\
\> $(\ell',r') \leftarrow (\rank_1(L,\ell),\rank_1(L,r))$ \\
\> \textbf{return} $\operatorname{list}(\ell', r')$ \\
\\
\textbf{function} $\operatorname{list}(\ell', r')$ \\
\> \textbf{if} $\ell'>r'$: \textbf{return} $\emptyset$ \\
\> $i' \leftarrow \rmq_\VILCP(\ell',r')$ \\
\> $i \leftarrow \max(\ell,\select(L,i'))$ \\
\> $j \leftarrow \min(r,\select(L,i'+1)-1)$ \\
\> $res \leftarrow \emptyset$ \\
\> \textbf{for} $k \leftarrow i\ldots j$: \\
\> \> $g \leftarrow \rank_1(B,\SA[k])$ \\
\> \> \textbf{if} $V[g] = 1$: \textbf{return} $res$ \\ 
\> \> $V[g] \leftarrow 1$ \\
\> \> $res \leftarrow res \cup \{ g \}$ \\
\> \textbf{return} $res \cup \operatorname{list}(\ell',i'-1)
			\cup \operatorname{list}(i'+1,r')$
\end{tabbing}
\end{minipage}
\caption{Pseudocode for document listing using the $\ILCP$ array.
Function $\operatorname{listDocuments}(\ell, r)$ lists the documents from interval $\SA[\ell..r]$; $\operatorname{list}(\ell',r')$ returns the distinct documents mentioned in the runs $\ell'$ to $r'$ that also belong to $\DA[\ell..r]$. We assume that in the beginning it holds $V[k]=0$ for all $k$; this can be arranged by resetting to 0 the same positions after the query or by using initializable arrays. All the unions on $res$ are known to be disjoint.} 
\label{fig:ilcp}
\end{figure}

We show next that this is correct as long as RMQ returns the
leftmost minimum in the range and that we recurse first to the left and then
to the right of each minimum $\VILCP[i']$ found.

\begin{lemma}
Using the procedure described, we correctly find all the positions $\ell \le
k\le r$ such that $\ILCP[k] < m$.
\end{lemma}
\begin{proof}
Let $j=\DA[k]$ be the leftmost occurrence of document $j$ in $\DA[\ell..r]$.
By Lemma~\ref{lem:leftmost}, among all the positions where $\DA[k']=j$ in
$\DA[\ell..r]$, $k$ is the only one where $\ILCP[k]<m$. Since we find a
minimum $\ILCP$ value in the range, and then explore the left subrange
before the right subrange, it is not possible to find first another occurrence
$\DA[k']=j$, since it has a larger $\ILCP$ value and is to the right of $k$.
Therefore, when $V[\DA[k]]=0$, that is, the first time we find a $\DA[k]=j$,
it must hold that $\ILCP[k]<m$,
and the same is true for all the other $\ILCP$ values in the run. Hence it
is correct to list all those documents and mark them in $V$. Conversely,
whenever we find a $V[\DA[k']]=1$, the document has already been reported. Thus
this is not its leftmost occurrence and then $\ILCP[k'] \ge m$ holds, as well
as for the whole run. Hence it is correct to avoid reporting the whole run and
to stop the recursion in the range, as the minimum value is already at least $m$.
\qed
\end{proof}

Note that we are not storing $\VILCP$ at all.
We have obtained our first result for document listing, where we recall that
$\rho$ is small on repetitive collections (Lemma~\ref{lem:rho}):

\begin{theorem} \label{thm:ILCP RMQ listing}
Let $T=S_1 \cdot S_2 \dotsm S_d$ be the concatenation of $d$ documents $S_j$,
and $\CSA$ be a compressed suffix array on $T$, searching for any pattern $P[1..m]$
in time $\search{m}$ and accessing $\SA[i]$ in time $\lookup{n}$.
Let $\rho$ be the number of runs in the $\ILCP$ array of $T$.
We can store $T$ in $\abs{\CSA} + \rho \log (n / \rho) + \Oh{\rho} + d \log (n
/ d) + \Oh{d} = \abs{\CSA} + \Oh{(\rho+d)\lg n}$ 
bits such that document listing takes $\Oh{\search{m} + \df
\cdot (\lookup{n}+\lg n)}$ time.
\end{theorem}

\subsection{Document Counting} \label{sec:ilcp-count}

Array $\ILCP$ also allows us to efficiently count the number of distinct
documents where $P$ appears, without listing them all. This time we will
explicitly represent $\VILCP$, in the following convenient way: consider 
a skewed wavelet tree (Section~\ref{sec:ranksel}), where the
leftmost leaf is at depth 1, the next 2 leaves are at depth 3, the next 4
leaves are at depth 5, and in general the $2^{d-1}$th to $(2^d-1)$th leftmost
leaves are at depth $2d-1$. Then the $i$th leftmost leaf is at depth
$1+2\lfloor \lg i \rfloor = \Oh{\log i}$. The number of wavelet tree nodes
up to depth $d$ is $\sum_{i=1}^{(d+1)/2} 2^i = 2(2^{(d+1)/2}-1)$. The
number of nodes up to the depth of the $m$th leftmost leaf is maximized
when $m$ is of the form $m=2^{d-1}$, reaching $2(2^d-1)=4m-2=\Oh{m}$.
See Figure~\ref{fig:wtree}.

\begin{figure}[t]
\centerline{\includegraphics[width=0.7\textwidth]{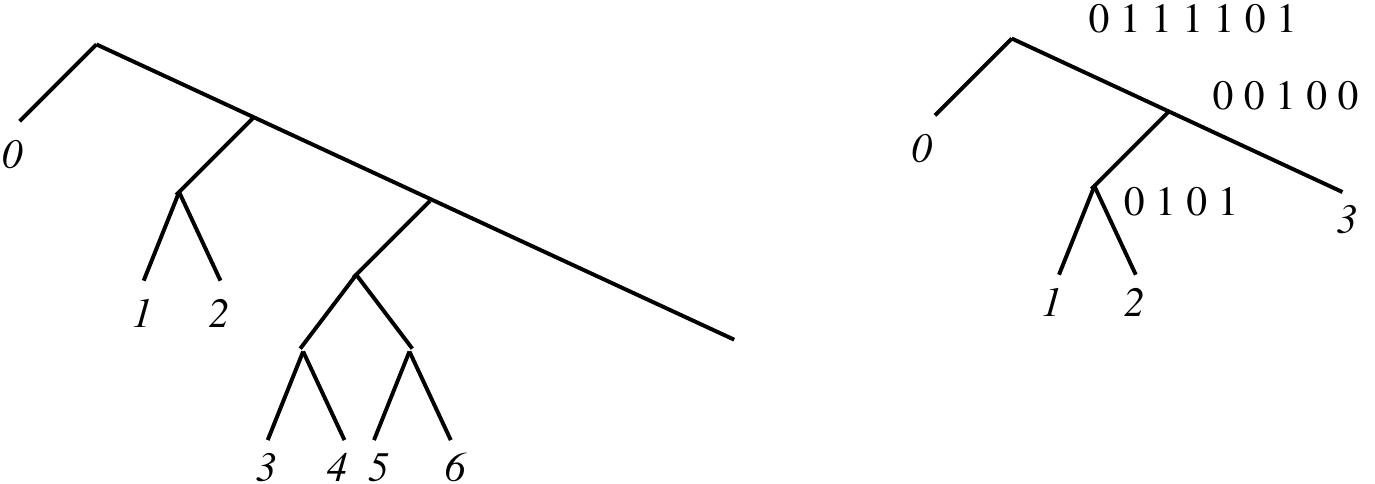}}
\caption{On the left, the schematic view of our skewed wavelet tree; on the
right, the case of our running example where it represents $\VILCP=\langle
0,1,2,3,1,0,2\rangle$.}
\label{fig:wtree}
\end{figure}

Let $\lambda$ be the maximum value in the $\ILCP$ array. Then the height of 
the wavelet tree is $\Oh{\lg \lambda}$ and the representation of $\VILCP$ 
takes at most $\rho\lg \lambda + o(\rho \lg \lambda)$ bits. 
If the documents $S$ are generated
using the A2 probabilistic model of \cite{Szp93}, then $\lambda=\Oh{\lg \abs{S}}=
\Oh{\lg n}$, and $\VILCP$ uses $\rho\lg\lg n (1+o(1))$ bits. The same happens 
under the model used in Section~\ref{sec:runs}.

The number of documents where $P$ appears, $\df$, is the number of times
a value smaller than $m$ occurs in $\ILCP[\ell..r]$. An algorithm to find
all those values in a wavelet tree of $\ILCP$ is as follows
\citep{GNP11}. Start at the
root with the range $[\ell..r]$ and its bitvector $W$. Go to the left child 
with the interval $[\rank_0(W,\ell-1)+1..\rank_0(W,r)]$ and to the right
child with the interval $[\rank_1(W,\ell-1)+1..\rank_1(W,r)]$, stopping the
recursion on empty intervals. This method arrives at all the wavelet tree
leaves corresponding to the distinct values in $\ILCP[\ell..r]$. Moreover,
if it arrives at a leaf $l$ with interval $[\ell_l..r_l]$, then there are
$r_l-\ell_l+1$ occurrences of the symbol of that leaf in $\ILCP[\ell..r]$.

Now, in the skewed wavelet tree of $\VILCP$, we are interested in the 
occurrences of symbols $0$ to $m-1$. Thus we apply the above algorithm but
we do not enter into subtrees handling an interval of values that is 
disjoint with $[0..m-1]$. Therefore, we only arrive at the $m$ leftmost leaves
of the wavelet tree, and thus traverse only $\Oh{m}$ wavelet tree nodes, in 
time $\Oh{m}$. 

A complication is that $\VILCP$ is the array of run length heads, so when
we start at $\VILCP[\ell'..r']$ and arrive at each leaf $l$ with interval 
$[\ell'_l..r'_l]$, we only know that $\VILCP[\ell'..r']$ contains from the 
$\ell'_l$th to the $r'_l$th occurrences of value $l$ in $\VILCP[\ell'..r']$. 
We store a reordering of
the run lengths so that the runs corresponding to each value $l$ are collected
left to right in $\ILCP$ and stored aligned to the wavelet tree leaf $l$.
Those are concatenated into another bitmap $L'[1..n]$ with $\rho$ 1s, similar
to $L$, which allows us, using $\select(L',\cdot)$, to count the total length
spanned by the $\ell'_l$th to $r'_l$th runs in leaf $l$. By adding the areas
spanned over the $m$ leaves, we count the total number of documents where $P$
occurs.  Note that we need to correct the lengths of runs $\ell'$ and $r'$,
as they may overlap the original interval $\ILCP[\ell..r]$.
Figure~\ref{fig:ilcp-count} gives the pseudocode.

\begin{figure}[t]\centering
\begin{minipage}[t]{0.5\linewidth}
\begin{tabbing}
mm\=mm\=mm\=mm\= \kill
\textbf{function} $\operatorname{countDocuments}(\ell, r)$ \\
\> $(\ell',r') \leftarrow (\rank_1(L,\ell),\rank_1(L,r))$ \\
\> $l \leftarrow m$ \\
\> $c\leftarrow\operatorname{count}(root,\ell', r')$ \\
\> \textbf{if} $\VILCP[\ell']<m$: $c \leftarrow c-(\ell-\select(L,\ell'))$\\
\> \textbf{if} $\VILCP[r']<m$: $c \leftarrow c-(\select(L,r'+1)-1-r)$\\
\> \textbf{return} $c$ \\
\\
\textbf{function} $\operatorname{count}(v,\ell', r')$ \\
\> \textbf{if} $l=0$: \textbf{return} $0$ \\
\> \textbf{if} $v$ is a leaf: \\
\>\> $l \leftarrow l-1$ \\
\>\> \textbf{if} $\ell'>r'$: \textbf{return} $0$ \\
\>\> \textbf{return} $\select(L',r'+1)-\select(L',\ell')$ \\
\> $(\ell_1,r_1) \leftarrow (\rank_1(v.W,\ell'-1)+1,\rank_1(v.W,r'))$ \\
\> \textbf{return} $\operatorname{count}(v.left,\ell'-\ell_1+1,r'-r_1) +
		    \operatorname{count}(v.right,\ell_1,r_1)$
\end{tabbing}
\end{minipage}
\caption{Document counting with the $\ILCP$ array.
Function $\operatorname{countDocuments}(\ell, r)$ counts the distinct documents from interval $\SA[\ell..r]$; $\operatorname{count}(v,\ell',r')$ returns the number of documents mentioned in the runs $\ell'$ to $r'$ under wavelet tree node $v$ that also belong to $\DA[\ell..r]$. We assume that the wavelet tree root node is $root$, and that any internal wavelet tree node $v$ has fields $v.W$ (bitvector), $v.left$ (left child), and $v.right$ (right child). Global variable $l$ is used to traverse the first $m$ leaves. The access to $\VILCP$ is also done with
the wavelet tree.}
\label{fig:ilcp-count}
\end{figure}

\begin{theorem} \label{thm:ILCP wavelet tree counting}
Let $T=S_1 \cdot S_2 \dotsm S_d$ be the concatenation of $d$ documents $S_j$,
and $\CSA$ a compressed suffix array on $T$ that searches for any pattern
$P[1..m]$ in time $\search{m}$.
Let $\rho$ be the number of runs in the $\ILCP$ array of $T$ and $\lambda$ be the
maximum length of a repeated substring inside any $S_j$.
Then we can store $T$ in \(\abs{\CSA} + \rho(\log \lambda + 2\log (n / \rho) +
\Oh{1}) = \abs{\CSA}+\Oh{\rho\lg n}\) bits such that the number of documents
where a pattern $P[1..m]$ occurs can be computed in time $\Oh{m+\search{m}}$.
\end{theorem}

\section{Precomputed Document Lists}
\label{sec:pdl}

In this section we introduce the idea of precomputing the answers of document 
retrieval queries for a sample of suffix tree nodes, and then
exploit repetitiveness by grammar-compressing the resulting sets of answers.
Such grammar compression is effective when the underlying collection is repetitive. The queries
are then extremely fast on the sampled nodes, whereas on the others we have a
way to bound the amount of work performed. The resulting structure is called
PDL (Precomputed Document Lists), for which we develop a variant for document 
listing and another for top-$k$ retrieval queries.

\subsection{Document Listing} \label{sec:pdl-list}

Let $v$ be a suffix tree node. We write $\SA_{v}$ to denote the interval of
the suffix array covered by node $v$, and $D_{v}$ to denote the set of
distinct document identifiers occurring in the same interval of the document
array. Given a block size $b$ and a constant $\beta \ge 1$, we build a sampled suffix tree that allows us to answer document listing queries efficiently. For any suffix tree node $v$, it holds that:
\begin{enumerate}
\item node $v$ is sampled and thus set $D_v$ is directly stored; or
\item $\abs{\SA_{v}} < b$, and thus documents can be listed in time $\Oh{b
\cdot \lookup{n}}$ by using a $\CSA$ and the bitvectors $B$ and $V$ of
Section~\ref{sec:muthu}; or
\item we can compute the set $D_{v}$ as the union of stored sets $D_{u_{1}}, \dotsc, D_{u_{k}}$ of total size at most $\beta \cdot \abs{D_{v}}$, where nodes $u_{1}, \dotsc, u_{k}$ are the children of $v$ in the sampled suffix tree.
\end{enumerate}

The purpose of rule 2 is to ensure that suffix array intervals solved by brute
force are not longer than $b$. The purpose of rule 3 is to ensure that, if we
have to rebuild an answer by merging a list of answers precomputed at descendant
sampled suffix tree nodes, then the merging costs no more than $\beta$ per result.
That is, we can discard answers of nodes that are close to being the union of
the answers of their descendant nodes, since we do not waste too much work in
performing the unions of those descendants. Instead, if the answers of the
descendants have many documents in common, then it is worth storing
the answer at the node too; otherwise merging will require much work because the
same document will be found many times (more than $\beta$ on average).

We start by selecting suffix tree nodes $v_{1}, \dotsc, v_{L}$, so that no
selected node is an ancestor of another, and the intervals $\SA_{v_{i}}$ of
the selected nodes cover the entire suffix array. Given node $v$ and its
parent $w$, we select $v$ if $\abs{\SA_{v}} \le b$ and $\abs{\SA_{w}} > b$,
and store $D_{v}$ with the node. These nodes $v$ become the leaves of the
sampled suffix tree, and we assume that they are numbered from left to right.
We then assume that all the ancestors of those leaves belong to the sampled
suffix tree, and proceed upward in the suffix tree removing some of them. Let $v$ be an internal node, $u_{1}, \dotsc, u_{k}$ its children, and $w$ its parent. If the total size of sets $D_{u_{1}}, \dotsc, D_{u_{k}}$ is at most $\beta \cdot \abs{D_{v}}$, we remove node $v$ from the tree, and add nodes $u_{1}, \dotsc, u_{k}$ to the children of node $w$. Otherwise we keep node $v$ in the sampled suffix tree, and store $D_{v}$ there.

When the document collection is repetitive, the document array $\DA[1..n]$ is also repetitive. This property has been used in the past to compress $\DA$ using grammars \citep{NPVjea13}. We can apply a similar idea on the $D_v$ sets stored at the sampled suffix tree nodes, since $D_v$ is a function of the range $\DA[\ell..r]$ that corresponds to node $v$.

Let $v_{1}, \dotsc, v_{L}$ be the leaf nodes and $v_{L+1}, \dotsc, v_{L+I}$ the internal nodes of the sampled suffix tree. We use grammar-based compression to replace frequent subsets in sets $D_{v_{1}}, \dotsc, D_{v_{L+I}}$ with grammar rules expanding to those subsets. Given a set $Z$ and a grammar rule $X \to Y$, where $Y \subseteq \set{1, \dotsc, d}$, we can replace $Z$ with $(Z \cup \set{X}) \setminus Y$, if $Y \subseteq Z$. As long as $\abs{Y} \ge 2$ for all grammar rules $X \to Y$, each set $D_{v_{i}}$ can be decompressed in $\Oh{\abs{D_{v_{i}}}}$ time.

To choose the replacements,
consider the bipartite graph with vertex sets $\set{v_{1}, \dotsc, v_{L+I}}$ and $\set{1, \dotsc, d}$, with an edge from $v_{i}$ to $j$ if $j \in D_{v_{i}}$. Let $X \to Y$ be a grammar rule, and let $V$ be the set of nodes $v_{i}$ such that rule $X \to Y$ can be applied to set $D_{v_{i}}$. As $Y \subseteq D_{v_{i}}$ for all $v_{i} \in V$, the induced subgraph with vertex sets $V$ and $Y$ is a complete bipartite graph or a \emph{biclique}. Many Web graph compression algorithms are based on finding bicliques or other dense subgraphs \citep{HNspire12.3}, and we can use these algorithms to find a good grammar compressing the precomputed document lists.

When all rules have been applied, we store the reduced sets $D_{v_{1}},
\dotsc, D_{v_{L+I}}$ as an array $A$ of document and rule
identifiers. The array takes $\abs{A} \log (d + n_{R})$ bits of space, where
$n_{R}$ is the total number of rules. We mark the first cell in the encoding
of each set with a $1$ in a bitvector $B_{A}[1..\abs{A}]$, so that set
$D_{v_{i}}$ can be retrieved by decompressing $A[\select(B_{A}, i)..\select(B_{A}, i+1) - 1]$. The bitvector takes $\abs{A}(1+o(1))$ bits of space
and answers $\select$ queries in $\Oh{1}$ time. The grammar rules
are stored similarly, in an array $G$ taking $\abs{G} \log d$ bits, with a
bitvector $B_{G}[1..\abs{G}]$ of $\abs{G}(1+o(1))$ bits separating the array
into rules (note that right hand sides of rules are formed only by terminals).

In addition to the sets and the grammar, we must also store the sampled
suffix tree. A bitvector $B_{L}[1..n]$ marks the first cell of interval
$\SA_{v_{i}}$ for all leaf nodes $v_{i}$, allowing us to convert interval
$\SA[\ell..r]$ into a range of nodes $[ln..rn] = [\rank_1(B_{L}, \ell) ..
\rank_1(B_{L}, r+1) - 1]$. Using the format of \cite{OS07} for $B_L$, the bitvector
takes $L \log (n/L) + \Oh{L}$ bits, and answers $\rank$ queries 
in $\Oh{\log (n/L)}$ time and $\select$ queries in constant time. A second 
bitvector $B_{F}[1..L+I]$, using $(L+I)(1+o(1))$ bits and supporting $\rank$
queries in constant time, marks the nodes that are the first 
children of their parents. An array $F[1..I]$ of $I \log I$ bits stores 
pointers from first children to their parent nodes, so that if node $v_{i}$ is 
a first child, its parent node is $v_{j}$, where $j = L + F[\rank_1(B_{F}, i)]$. 
Finally, array $N[1..I]$ of $I \log L$ bits stores a pointer to the leaf node 
following those below each internal node. 

\begin{figure}[t]\centering
\begin{minipage}[t]{0.5\linewidth}
\begin{tabbing}
mm\=mm\=mm\=mm\= \kill
\textbf{function} $\operatorname{listDocuments}(\ell, r)$ \\
\> $(res, ln) \leftarrow (\emptyset, \rank_1(B_{L}, \ell))$ \\
\> \textbf{if} $\select(B_{L}, ln) < \ell$: \\
\> \> $r' \leftarrow \min(\select(B_{L}, ln+1) - 1, r)~~~$ \\
\> \> $(res, ln) \leftarrow (\operatorname{list}(\ell, r'), ln + 1)$ \\
\> \> \textbf{if} $r' = r$: \textbf{return} $res$ \\
\> $rn \leftarrow \rank_1(B_{L}, r+1) - 1$ \\
\> \textbf{if} $\select(B_{L}, rn+1) \le r$: \\
\> \> $\ell' \leftarrow \select(B_{L}, rn+1)$ \\
\> \> $res \leftarrow res \cup \operatorname{list}(\ell', r)$ \\
\> \textbf{return} $res \cup \operatorname{decompress}(ln, rn)$ \\
\\
\textbf{function} $\operatorname{decompress}(\ell, r)$ \\
\> $(res, i) \leftarrow (\emptyset, \ell)$ \\
\> \textbf{while} $i \le r$: \\
\> \> $next \leftarrow i + 1$ \\
\> \> \textbf{while} $B_{F}[i] = 1$: \\ 
\> \> \> $(i', next') \leftarrow \operatorname{parent}(i)$ \\
\> \> \> \textbf{if} $next' > r + 1$: \textbf{break} \\
\> \> \> $(i, next) \leftarrow (i', next')$ \\
\> \> $res \leftarrow res \cup \operatorname{set}(i)$ \\
\> \> $i \leftarrow next$ \\
\> \textbf{return} $res$
\end{tabbing}
\end{minipage}
\begin{minipage}[t]{.5\linewidth}
\begin{tabbing}
mm\=mm\=mm\=mm\= \kill
\textbf{function} $\operatorname{parent}(i)$\\
\> $par \leftarrow F[\rank_1(B_{F}, i)]$\\
\> \textbf{return} $(par + L, N[par])$\\
\\
\textbf{function} $\operatorname{set}(i)$\\
\> $res \leftarrow \emptyset$\\
\> $\ell \leftarrow \select(B_{A}, i)$\\
\> $r \leftarrow \select(B_{A}, i+1) - 1$\\
\> \textbf{for} $j \leftarrow \ell$ \textbf{to} $r$:\\
\> \> \textbf{if} $A[j] \le d$: $res \leftarrow res \cup \set{A[j]}$\\
\> \> \textbf{else}: $res \leftarrow res \cup \operatorname{rule}(A[j] - d)$\\
\> \textbf{return} $res$\\
\\
\textbf{function} $\operatorname{rule}(i)$\\
\> $\ell \leftarrow \select(B_{G}, i)$\\
\> $r \leftarrow \select(B_{G}, i+1) - 1$\\
\> \textbf{return} $G[\ell..r]$\\
\\
\textbf{function} $\operatorname{list}(\ell, r)$\\
\> $res \leftarrow \emptyset$ \\
\> \textbf{for} $i \leftarrow \ell$ \textbf{to} $r$:\\
\> \> $res \leftarrow res \cup \set{\rank_1(B, \SA[i])}$\\
\> \textbf{return} $res$
\end{tabbing}
\end{minipage}
\caption{Document listing using precomputed answers.
Function $\operatorname{listDocuments}(\ell, r)$ lists the documents from interval $\SA[\ell..r]$; $\operatorname{decompress}(\ell, r)$ decompresses the sets stored in nodes $v_{\ell}, \dotsc, v_{r}$; $\operatorname{parent}(i)$ returns the parent node and the leaf node following it for a first child $v_{i}$; $\operatorname{set}(i)$ decompresses the set stored in $v_{i}$;  $\operatorname{rule}(i)$ expands the $i$th grammar rule; and $\operatorname{list}(\ell, r)$ lists the documents from interval $\SA[\ell..r]$ by using $\CSA$ and bitvector $B$.}
\label{fig:algorithm}
\end{figure}

Figure~\ref{fig:algorithm} gives the pseudocode for document listing using the
precomputed answers. Function $\operatorname{list}(\ell, r)$ takes $\Oh{(r + 1
- \ell)\,\lookup{n}}$ time, $\operatorname{set}(i)$ takes
$\Oh{\abs{D_{v_{i}}}}$ time, and $\operatorname{parent}(i)$ takes $\Oh{1}$
time. Function $\operatorname{decompress}(\ell, r)$ produces set $res$ in time
$\Oh{\abs{res}\cdot \beta h}$, where $h$ is the height of the sampled suffix
tree: finding each set may take $\Oh{h}$ time, and we may encounter the same
document $\Oh{\beta}$ times. Hence the total time for $\operatorname{listDocuments}(\ell, r)$
is $\Oh{\df \cdot \beta h + \log n}$ for unions of precomputed answers, and
$\Oh{b \cdot \lookup{n}}$ otherwise. If the text follows the A2 model of
\citet{Szp93}, then $h=\Oh{\log n}$ and the total time is on average
$\Oh{\df \cdot \beta\log n + b \cdot \lookup{n}}$.

We do not write the result as a theorem because we cannot upper bound
the space used by the structure in terms of $b$ and $\beta$. In a bad case
like $T=a^{\ell-1}\$b^{\ell-1}\$c^{\ell-1}\$\ldots$, the suffix tree is 
formed by $d$ long paths and the sampled suffix tree contains at least $d(n/d-b)
= \Theta(n)$ nodes (assuming $bd = o(n)$), so the total space is $\Oh{n\log n}$
bits as in a classical suffix tree. In a good case, such as a balanced suffix 
tree
(which also arises on texts following the A2 model), 
the sampled suffix tree has
$\Oh{n/b}$ nodes. Although each such node $v$ may store a list $D_v$ with $b$ 
entries, many of those entries are similar when the collection is repetitive, 
and thus their compression is effective.

\subsection{Top-$k$ Retrieval}
\label{sec:topk}

Since we have the freedom to represent the documents in sets $D_v$ in any order,
we can in particular sort the document identifiers in decreasing order of their
``frequencies'', that is, the number of times the string represented by $v$ 
appears in the
documents. Ties are broken by document identifiers in increasing order. Then a 
top-$k$ query on a node $v$ that stores its list $D_v$ boils down to listing 
the first $k$ elements of $D_v$.

This time we cannot use the set-based grammar compressor, but we need, instead,
a compressor that preserves the order. We use Re-Pair \citep{LM00}, which
produces a grammar where each nonterminal produces two new symbols, 
terminal or nonterminal. As Re-Pair decompression is recursive, decompression
can be slower than in document listing, although it is still fast in practice and takes linear
time in the length of the decompressed sequence.

In order to merge the results from multiple nodes in the sampled suffix tree, we
need to store the frequency of each document. These are stored in the same 
order as the identifiers. Since the frequencies are nonincreasing, with
potentially long runs of small values, we can represent them space-efficiently 
by run-length encoding the sequences and using differential 
encoding for the run heads. A node containing $s$ suffixes in its subtree
has at most $\Oh{\sqrt{s}}$ distinct frequencies, and the frequencies can be
encoded in $\Oh{\sqrt{s} \log s}$ bits.

There are two basic approaches to using the PDL structure for top-$k$ document
retrieval. First, we can store the document lists for all suffix tree nodes above the leaf blocks, producing a structure that is essentially an inverted index for all frequent substrings. This approach is very fast, as we need only decompress the first $k$ document identifiers from the stored sequence, and it works well with repetitive collections thanks to the grammar-compression of the lists. 
Note that this enables {\em incremental top-$k$ queries}, where value $k$ is
not given beforehand, but we extract documents with successively lower scores
and can stop at any time. Note also that, in this version, it is not necessary
to store the frequencies.

Alternatively, we can build the PDL structure as in
Section~\ref{sec:pdl-list}, with some parameter $\beta$, to achieve better space usage. Answering queries will then be slower as we have to decompress multiple document sets, merge the sets, and determine the top $k$ documents.
We tried different heuristics for merging prefixes of the document sequences, stopping when a correct answer to the top-$k$ query could be guaranteed. The heuristics did not generally work well, making brute-force merging the fastest alternative. 

\section{Engineering a Document Counting Structure}
\label{sec:count}

In this section we revisit a generic document counting structure by
\cite{Sad07}, which uses $2n+o(n)$ bits and answers counting queries in 
constant time. We show that
the structure inherits the repetitiveness present in the text collection, 
which can then be exploited to reduce its space occupancy. Surprisingly, the 
structure also becomes repetitive with random and near-random data, such 
as unrelated DNA sequences, which is a result of interest for general string collections.
We show how to take advantage of this redundancy in a number of different ways,
leading to different time/space trade-offs.

\subsection{The Basic Bitvector}

We describe the original document structure of \cite{Sad07}, which computes
$\df$ in constant time given the locus of the pattern $P$ (i.e., the suffix
tree node arrived at when searching for $P$), while using just
$2n+o(n)$ bits of space. 

We start with the suffix tree of the text, and add new internal nodes to it to make it a binary tree. For each internal node $v$ of the binary suffix tree, let
$D_v$ be again the set of distinct document identifiers in the corresponding
range $\DA[\ell..r]$, and let $\countq(v)=|D_v|$ be the size of that set. If node
$v$ has children $u$ and $w$, we define the number of {\em redundant} suffixes as
$h(v) = \abs{D_u \cap D_w}$. This allows us to compute $\df$ recursively: $\countq(v) = \countq(u) + \countq(w) - h(v)$. By using the leaf nodes descending from $v$, 
$[\ell..r]$, as base cases, we can solve the recurrence:
\begin{displaymath}
\countq(v) = \countq(\ell,r) = (r + 1 - \ell) - \sum_{u} h(u),
\end{displaymath}
where the summation goes over the internal nodes of the subtree rooted at $v$.

We form an array $H[1..n-1]$ by traversing the internal nodes in inorder and listing the $h(v)$ values. As the nodes are listed in inorder, subtrees form contiguous ranges in the array. We can therefore rewrite the solution as
\begin{displaymath}
\countq(\ell,r) = (r + 1 - \ell) - \sum_{i=\ell}^{r-1} H[i].
\end{displaymath}
To speed up the computation, we encode the array in unary as bitvector $H'$. Each cell $H[i]$ is encoded as a 1-bit, followed by $H[i]$ 0s. We can now compute the sum by counting the number of 0s between the 1s of ranks $\ell$ and $r$:
\begin{displaymath}
\countq(\ell,r) = 2(r - \ell) - (\select_{1}(H',r) - \select_{1}(H',\ell)) + 1.
\end{displaymath}
As there are $n-1$ 1s and $n-d$ 0s, bitvector $H'$ takes at most $2n+o(n)$ bits.

\subsection{Compressing the Bitvector}
\label{section:new}

The original bitvector requires $2n + o(n)$ bits,
regardless of the underlying data. This can be a considerable
overhead with highly compressible collections, taking significantly more space
than the $\CSA$ (on top of which the structure operates).
Fortunately, as we now show, the bitvector $H'$ used in \citeauthor{Sad07}'s method is highly compressible. There are five main ways of compressing the bitvector, with different combinations of them working better with different datasets.

\begin{enumerate}

\item Let $V_{v}$ be the set of nodes of the binary suffix tree corresponding to node $v$ of the original suffix tree. As we only need to compute $\countq()$ for the nodes of the original suffix tree, the individual values of $h(u)$, $u \in V_{v}$, do not matter, as long as the sum $\sum_{u \in V_{v}} h(u)$ remains the same. We can therefore make bitvector $H'$ more compressible by setting $H[i] = \sum_{u \in V_{v}} h(u)$, where $i$ is the inorder rank of node $v$, and $H[j] = 0$ for the rest of the nodes. As there are no real drawbacks in this reordering, we will use it with all of our variants of Sadakane's method.

\item \emph{Run-length encoding} works well with versioned collections and
collections of random documents. When a pattern occurs in many documents, but no more than once in each, the corresponding subtree will be encoded as a run of 1s in $H'$.

\item When the documents in the collection have a versioned structure, we can reasonably expect \emph{grammar compression} to be effective. To see this, consider a substring $x$ that occurs in many documents, but at most once in each document. If each occurrence of substring $x$ is preceded by symbol $a$, the subtrees of the binary suffix tree corresponding to patterns $x$ and $ax$ have an identical structure, and the corresponding areas in $D$ are identical. Hence the subtrees are encoded identically in bitvector $H'$.

\item If the documents are internally repetitive but unrelated to each other,
the suffix tree has many subtrees with suffixes from just one document. We can
prune these subtrees into leaves in the binary suffix tree, using a
\emph{filter} bitvector $F[1..n-1]$ to mark the remaining nodes. Let $v$ be a
node of the binary suffix tree with inorder rank $i$. We will set $F[i] = 1$
iff $\countq(v) > 1$. Given a range $[\ell..r-1]$ of nodes in the binary suffix tree, the corresponding subtree of the pruned tree is $[\rank_{1}(F,\ell) .. \rank_{1}(F, r-1)]$. The filtered structure consists of bitvector $H'$ for the pruned tree and a compressed encoding of $F$.

\item We can also use filters based on the values in  array $H$ instead of the sizes of the document sets. If $H[i] = 0$ for most cells, we can use a \emph{sparse filter} $F_{S}[1..n-1]$, where $F_{S}[i] = 1$ iff $H[i] > 0$, and build bitvector $H'$ only for those nodes. We can also encode positions with $H[i] = 1$ separately with a \emph{$1$-filter} $F_{1}[1..n-1]$, where $F_{1}[i] = 1$ iff $H[i] = 1$. With a $1$-filter, we do not write 0s in $H'$ for nodes with $H[i] = 1$, but instead subtract the number of 1s in $F_{1}[\ell..r-1]$ from the result of the query. It is also possible to use a sparse filter and a $1$-filter simultaneously. In that case, we set $F_{S}[i] = 1$ iff $H[i] > 1$.

\end{enumerate}

\subsection{Analysis}

We analyze the number of runs of 1s in bitvector $H'$ in the expected
case. Assume that our document collection consists of $d$ documents,
each of length $r$, over an alphabet of size $\sigma$. We call string $S$ \emph{unique},
if it occurs at most once in every document. The subtree of the binary suffix tree
corresponding to a unique string is encoded as a run of 1s in
bitvector $H'$. If we can cover all leaves of the tree with $u$ unique substrings, bitvector
$H'$ has at most $2u$ runs of 1s.

Consider a random string of length $k$. Suppose the probability that the
string occurs at least twice in a given document is at most
$r^2 / (2 \sigma^{2 k})$, which is the case if, e.g., we choose each document
randomly or we choose one document randomly and generate the others by copying
it and randomly substituting some symbols.  By the union bound, the probability the string is non-unique is at most
$d r^2 / (2 \sigma^{2 k})$.
Let $N(i)$ be the number of non-unique strings of length
$k_{i} = \log_{\sigma} (r \sqrt{d}) + i$. As there are $\sigma^{k_{i}}$ strings of length
$k_{i}$, the expected value of $N(i)$ is at most $r \sqrt{d} / (2 \sigma^{i})$. The expected
size of the smallest cover of unique strings is therefore at most
\begin{displaymath}
(\sigma^{k_{0}} - N(0)) + \sum_{i=1}^{\infty} (\sigma N(i-1) - N(i)) =
r \sqrt{d} + (\sigma - 1) \sum_{i=0}^{\infty} N(i) \le
\left( \frac{\sigma}{2} + 1 \right) r \sqrt{d},
\end{displaymath}
where $\sigma N(i-1) - N(i)$ is the number of strings that become unique at length $k_{i}$.
The number of runs of 1s in $H'$ is therefore sublinear in the size of the
collection ($dr$). 
See Figure~\ref{figure:runs} for an experimental confirmation of this analysis.

\begin{figure}[t]
\centering
\includegraphics[width=0.8\textwidth]{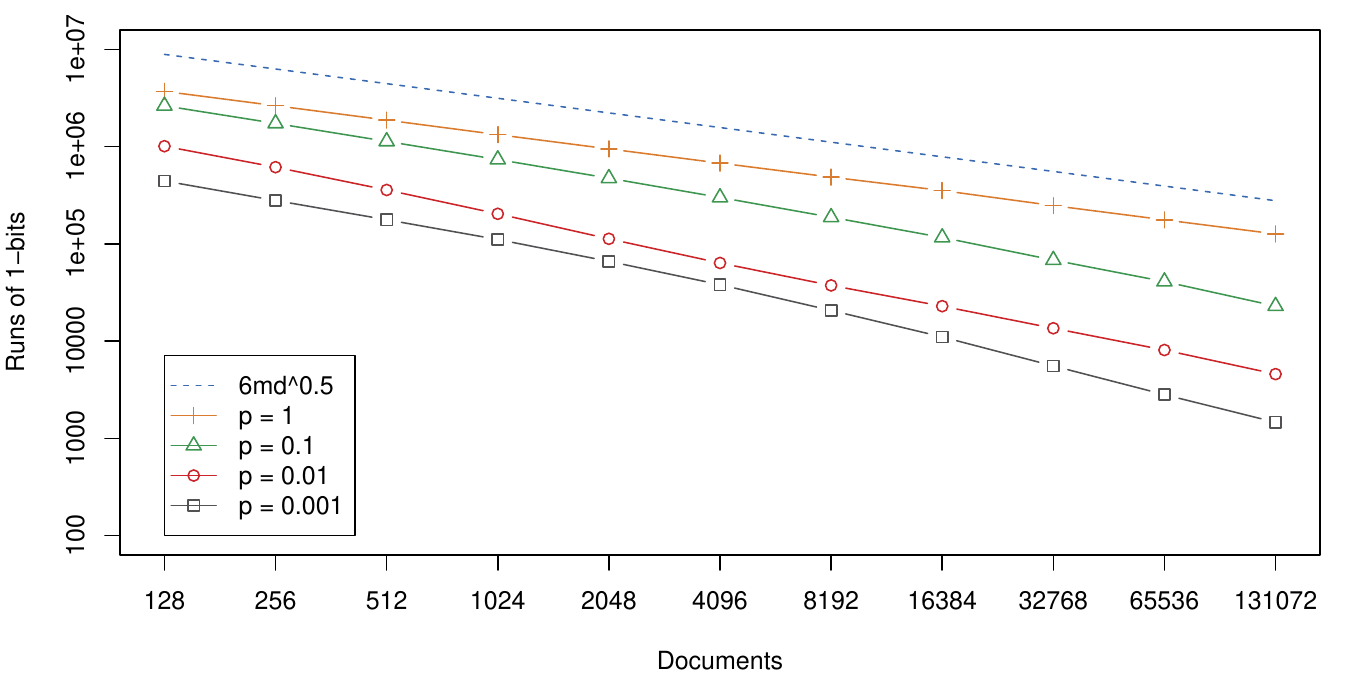}

\caption{The number of runs of $1$-bits in Sadakane's bitvector $H'$ on synthetic collections of DNA sequences ($\sigma = 4$). Each collection has been generated by taking a random sequence of length $m = 2^{7}$ to $2^{17}$, duplicating it $d = 2^{17}$ to $2^{7}$ times (making the total size of the collection $2^{24}$), and mutating the sequences with random point mutations at probability $p = 0.001$ to $1$. The mutations preserve zero-order empirical entropy by replacing the mutated symbol with a randomly chosen symbol according to the distribution in the original sequence. The dashed line represents the expected case upper bound for $p = 1$.}\label{figure:runs}
\end{figure}

\section{A Multi-term Index}
\label{sec:tfidf}

The queries we defined in the Introduction are {\em single-term}, that is, the 
query pattern $P$ is a single string. In this section we show how our indexes 
for single-term retrieval can be used for {\em ranked multi-term} queries on
repetitive text collections. The key idea is to regard our incremental top-$k$
algorithm of Section~\ref{sec:topk} as an {\em abstract representation} of the
inverted lists of the individual query terms, sorted by decreasing weight, and
then apply any algorithm that traverses those lists sequentially. Since our
relevance score will depend on the term frequency and the document frequency
of the terms, we will integrate a document counting structure as well
(Sections~\ref{sec:ilcp-count} or \ref{sec:count}).

Let $Q = \langle q_{1}, \dots, q_{m} \rangle$ be a query consisting of $m$
patterns $q_{i}$. We support ranked queries, which return the $k$ documents
with the highest scores among the documents matching the query. A
\emph{disjunctive} or \emph{ranked-OR} query matches document $D$ if at least
one of the patterns occurs in it, while a \emph{conjunctive} or
\emph{ranked-AND} query matches $D$ if all query patterns occur in it. Our index supports both conjunctive and disjunctive queries with \tfidf-like scores
$$
w(D, Q) = \sum_{i=1}^{m} w(D, q_{i}) = \sum_{i=1}^{m} f(\tf(D, q_{i})) \cdot g(\df(q_{i})),
$$
where $f \ge 0$ is an increasing function, $\tf(D, q_{i})$ is the \emph{term
frequency} (the number of occurrences) of pattern $q_{i}$ in document $D$, $g
\ge 0$ is a decreasing function, and $\df(q_{i})$ is the \emph{document
frequency} of pattern $q_{i}$. For example, the standard \tfidf\ scoring scheme corresponds to using
$f(\tf) = \tf$ and $g(\df) = \log (d / \max(\df, 1))$.

From Section~\ref{sec:topk}, we use the incremental variant, which stores the 
full answers for all the suffix tree nodes above leaves. 
The query algorithm uses $\CSA$ to find the lexicographic range
$[\ell_{i}..r_{i}]$ matching each pattern $q_{i}$. We then use \PDL{} to find
the sparse suffix tree node $v_{i}$ corresponding to range
$[\ell_{i}..r_{i}]$ and fetch its list $D_{v_i}$, which is stored in decreasing term frequency
order. If $v_i$ is not in the sparse suffix tree, we use instead the $\CSA$ to 
build $D_{v_i}$ by brute force from $\SA[\ell_{i}..r_{i}]$.
We also compute $\df(q_{i})=\countq(v_i)$ for all query
patterns $q_{i}$ with our document counting structure. The algorithm then iterates the following loop with $k' = 2k, 4k, 8k, \dotsc$:
\begin{enumerate}

\item Extract $k'$ more documents from the document list of $v_{i}$ for each
pattern $q_{i}$.

\item If the query is conjunctive, filter out extracted documents that do not
match the query patterns with completely decompressed document lists.

\item Determine a lower bound for $w(D, Q)$ for all documents $D$ extracted so far. If document $D$ has not been encountered in the document list of $v_{i}$, use $0$ as a lower bound for $w(D, q_{i})$.

\item Determine an upper bound for $w(D, Q)$ for all documents $D$. If
document $D$ has not been encountered in the document list of $v_{i}$, use
$\tf(D', q_{i})$, where $D'$ is the next unextracted document for pattern $q_{i}$, as an upper bound for $\tf(D, q_{i})$.

\item If the query is disjunctive, filter out extracted documents $D$ with smaller upper bounds for $w(D, Q)$ than the lower bounds for the current top-$k$ documents. Stop if the top-$k$ set cannot change further.

\item If the query is conjunctive, stop if the top-$k$ documents match all
query patterns and the upper bounds for the remaining documents are lower than the lower bounds for the top-$k$ documents.

\end{enumerate}
The algorithm always finds a correct top-$k$ set, although the scores may be incorrect if a disjunctive query stops early.

\section{Experiments and Discussion}
\label{sec:exp}

\subsection{Experimental Setup}

\subsubsection{Document Collections}

We performed extensive experiments with both real and synthetic
collections.\footnote{See \url{http://jltsiren.kapsi.fi/rlcsa}
for the datasets and full results.}
Most of our document collections were relatively small, around 100~MB in size, as some of the implementations \citep{NPVjea13} use 32-bit libraries. We also used larger versions of some collections, up to 1~GB in size, to see how the collection size affects the results. In general, collection size is more important in top-$k$ document retrieval. Increasing the number of documents generally increases the $\df/k$ ratio, and thus makes brute-force solutions based on document listing less appealing. In document listing, the size of the documents is more important than collection size, as a large $\occ/\df$ ratio makes brute-force solutions based on pattern matching less appealing.

The performance of various solutions depends both on the repetitiveness of the collection
and the type of the repetitiveness. Hence we used a fair number of real and synthetic collections with different characteristics for our experiments. We describe
them next, and summarize their statistics in Table~\ref{table:collections}.

\begin{sidewaystable}[p]
\centering
\caption{Statistics for document collections (small, medium, and large variants). Collection size, RLCSA size without suffix array samples, number of documents, average document length, number of patterns, average number of occurrences and document occurrences, and the ratio of occurrences to document occurrences. For the synthetic collections (second group), most of the statistics vary greatly.}\label{table:collections}

\begin{tabular}{lcccccccc}
\hline
\noalign{\smallskip}
Collection & \multicolumn{1}{c}{Size} & \multicolumn{1}{c}{CSA size} & \multicolumn{1}{c}{Documents} & \multicolumn{1}{c}{Avg.\ doc size} & \multicolumn{1}{c}{Patterns} & \multicolumn{1}{c}{Occurrences} & \multicolumn{1}{c}{Document occs} & \multicolumn{1}{c}{Occs per doc} \\
	   & \multicolumn{1}{c}{($n$)} & \multicolumn{1}{c}{(RLCSA)} & \multicolumn{1}{c}{($d$)} & \multicolumn{1}{c}{($n/d$)} & & \multicolumn{1}{c}{($\avg{\occ}$)} & \multicolumn{1}{c}{($\avg{\df}$)} & \multicolumn{1}{c}{($\avg{\occ/\df}$)} \\
\noalign{\smallskip}
\hline
\noalign{\smallskip}
\Page      &  110 MB &   2.58 MB &       60 & 1{,}919{,}382 &  7{,}658 &   781 &     3 & 242.75 \\
           & 641 MB  &   9.00 MB &      190 & 3{,}534{,}921 & 14{,}286 & 2{,}601 &     6 & 444.79 \\
           & 1037 MB &  17.45 MB &      280 & 3{,}883{,}145 & 20{,}536 &  2{,}889 &     7 & 429.04 \\
\noalign{\smallskip}
\Revision  &  110 MB &   2.59 MB &     8{,}834 &   13{,}005 &  7{,}658 &   776 &   371 &   2.09 \\
           &  640 MB &   9.04 MB &    31{,}208 &   21{,}490 & 14{,}284 & 2{,}592 &  1{,}065 &   2.43 \\
           & 1035 MB &  17.55 MB &    65{,}565 &   16{,}552 & 20{,}536 &  2{,}876 &  1{,}188 &   2.42 \\
\noalign{\smallskip}
\Enwiki    &  113 MB &  49.44 MB &     7{,}000 &   16{,}932 & 18{,}935 &  1{,}904 &   505 &   3.77 \\
           &  639 MB & 309.31 MB &    44{,}000 &   15{,}236 & 19{,}628 & 10{,}316 &  2{,}856 &   3.61 \\
           & 1034 MB & 482.16 MB &    90{,}000 &   12{,}050 & 19{,}805 & 17{,}092 &  4{,}976 &   3.44 \\
\noalign{\smallskip}
\Influenza &  137 MB &   5.52 MB &   100{,}000 &    1{,}436 &  1{,}000 & 24{,}975 & 18{,}547 &   1.35 \\
           &  321 MB &  10.53 MB &   227{,}356 &    1{,}480 &  1{,}000 & 59{,}997 & 44{,}012 &   1.36 \\
\noalign{\smallskip}
\Swissprot &   54 MB &  25.19 MB &   143{,}244 &        398 & 10{,}000 &   160 &   121 &   1.33 \\
\noalign{\smallskip}
\Wiki      & 1432 MB &  42.90 MB &   103{,}190 &   14{,}540 \\
\noalign{\smallskip}
\hline
\noalign{\smallskip}
\DNA       &   95 MB &           &   100{,}000 &            & 889--1{,}000 \\
\noalign{\smallskip}
\Concat    &   95 MB &           & 10--1{,}000 &            & 7{,}538--15{,}272 \\
\noalign{\smallskip}
\Version   &   95 MB &           &    10{,}000 &            & 7{,}537--15{,}271 \\
\noalign{\smallskip}
\hline
\end{tabular}
\end{sidewaystable}

\paragraph{A note on collection size.}
The index structures evaluated in this paper should be understood as promising
algorithmic ideas. In most implementations, the construction algorithms do not
scale up for collections larger than a couple of gigabytes. This is often
intentional. In this line of research, being able to easily evaluate
variations of the fundamental idea is more important than the speed or memory
usage of construction. As a result, many of the construction algorithms build an explicit suffix tree for the collection and store various kinds of additional information in the nodes. Better construction algorithms can be designed once the most promising ideas have been identified. See Appendix~\ref{appendix:construction} for further discussion on index construction.

\paragraph{Real collections.} 
We use various document collections from real-life repetitive scenarios. Some
collections come in small, medium, and large variants.
\Page{} and \Revision{} are repetitive collections generated from a Finnish-language Wikipedia archive with full version history. There are $60$ (small), $190$ (medium), or $280$ (large) pages with a total of $8{,}834$, $31{,}208$, or $65{,}565$ revisions. In $\Page$, all the revisions of a page form a single document, while each revision becomes a separate document in $\Revision$.
\Enwiki{} is a non-repetitive collection of $7{,}000$, $44{,}000$, or $90{,}000$ pages from a snapshot of the English-language Wikipedia.
\Influenza{} is a repetitive collection containing $100{,}000$ or $227{,}356$ sequences from influenza virus genomes (we only have small and large variants).
\Swissprot{} is a non-repetitive collection of $143{,}244$ protein sequences used
in many document retrieval papers (e.g., \cite{NPVjea13}). As the full collection is only 54~MB, only the small version of \Swissprot\ exists.
\Wiki{} is a repetitive collection similar to \Revision. It is generated by sampling all revisions of 1\% of pages from the English-language versions of Wikibooks, Wikinews, Wikiquote, and Wikivoyage.

\paragraph{Synthetic collections.}
To explore the effect of
collection repetitiveness on document retrieval performance in more detail, we generated three types of
synthetic collections, using files from the Pizza \& Chili corpus%
\footnote{\url{http://pizzachili.dcc.uchile.cl}}.
\DNA{} is similar to \Influenza. Each collection has $d=1$, $10$, $100$, or $1{,}000$ base documents, $100{,}000/d$ variants of each base document, and mutation rate $p = 0.001$, $0.003$, $0.01$, $0.03$, or $0.1$. We take a prefix of length $1{,}000$ from the Pizza \& Chili DNA file and generate the base documents by mutating the prefix at probability $10p$ under the same model as in Figure~\ref{figure:runs}. We then generate the variants in the same way with mutation rate $p$.
\Concat{} and \Version{} are similar to \Page{} and \Revision{}, respectively. We read $d=10$, $100$, or $1{,}000$ base documents of length $10{,}000$ from the Pizza \& Chili English file, and generate $10{,}000/d$ variants of each base document with mutation rates $0.001$, $0.003$, $0.01$, $0.03$, and $0.1$, as above. Each variant becomes a separate document in \Version{}, while all variants of the same base document are concatenated into a single document in \Concat.

\subsubsection{Queries}

\paragraph{Real collections.}
For \Page{} and \Revision{}, we downloaded a list of Finnish words from the Institute for the Languages in Finland, and chose all words of length $\ge 5$ that occur in the collection.
For \Enwiki{}, we used search terms from an MSN query log with stopwords filtered out. We generated $20{,}000$ patterns according to term frequencies, and selected those that occur in the collection.
For \Influenza{}, we extracted $100{,}000$ random substrings of length $7$, filtered out duplicates, and kept the $1{,}000$ patterns with the largest $\occ/\df$ ratios.
For \Swissprot{}, we extracted $200{,}000$ random substrings of length $5$, filtered out duplicates, and kept the $10{,}000$ patterns with the largest $\occ/\df$ ratios.
For \Wiki{}, we used the TREC 2006 Terabyte Track efficiency queries\footnote{\url{http://trec.nist.gov/data/terabyte06.html}} consisting of $411{,}394$ terms in $100{,}000$ queries.

\paragraph{Synthetic collections.}
We generated the patterns for \DNA{} with a similar process as for \Influenza{} and \Swissprot. We extracted $100{,}000$ substrings of length $7$, filtered out duplicates, and chose the $1{,}000$ with the largest $\occ/\df$ ratios.
For \Concat{} and \Version{}, patterns were generated from the MSN query log in the same way as for \Enwiki.

\subsubsection{Test Environment}

We used two separate systems for the experiments. For document listing and document counting, our test environment had two 2.40~GHz quad-core Intel Xeon E5620 processors and 96~GB memory. Only one core was used for the queries. The operating system was Ubuntu 12.04 with Linux kernel 3.2.0. All code was written in C++. We used g++ version 4.6.3 for the document listing experiments and version 4.8.1 for the document counting experiments.

For the top-$k$ retrieval and \tfidf\ experiments, we used another system with two 16-core AMD Opteron 6378 processors and 256~GB memory. We used only a single core for the single-term queries and up to 32~cores for the multi-term queries. The operating system was Ubuntu 12.04 with Linux kernel 3.2.0. All code was written in C++ and compiled with g++ version 4.9.2.

We executed the query benchmarks in the following way:
\begin{enumerate}

\item Load the RLCSA with the desired sample period for the current collection into memory.

\item Load the query patterns corresponding to the collection into memory and execute $\findq$ queries in the RLCSA. Store the resulting lexicographic ranges $[\ell..r]$ in vector $V$.

\item Load the index to be benchmarked into memory.

\item Iterate through vector $V$ once using a single thread and execute the desired query for each range $[\ell..r]$. Measure the total wall clock time for executing the queries.\label{step:queries}

\end{enumerate}

We divided the measured time by the number of patterns, and listed the average time per query in milliseconds or microseconds and the size of the index structure in bits per symbol. There were certain exceptions:
\begin{itemize}

\item \LZ{} and \Grammar{} do not use a \CSA. With them, we iterated through the vector of patterns as in step~\ref{step:queries}, once the index and the patterns had been loaded into memory. The average time required to get the range $[\ell..r]$ in \CSA-based indexes (4 to 6~microseconds, depending on the collection) was negligible compared to the average query times of \LZ{} (at least 170~microseconds) and \Grammar{} (at least 760~microseconds).

\item We used the existing benchmark code with \SURF{}. The code first loads the index into memory and then iterates through the pattern file by reading one line at a time. To reduce the overhead from reading the patterns, we cached them by using \texttt{cat > /dev/null}. Because \SURF{} queries were based on the pattern instead of the corresponding range $[\ell..r]$, we executed $\findq$ queries first and subtracted the time used for them from the subsequent top-$k$ queries.

\item In our \tfidf{} index, we parallelized step~\ref{step:queries} using the OpenMP \texttt{parallel for} construct.

\item We used the existing benchmark code with Terrier. We cached the queries as with \SURF, set \texttt{trec.querying.outputformat} to \texttt{NullOutputFormat}, and set the logging level to \texttt{off}.

\end{itemize}

\subsection{Document Listing}\label{sec:doclist-experiments}

We compare our new proposals from Sections~\ref{sec:repet} and \ref{sec:pdl-list}
to the existing document listing solutions. We also aim to determine when these
sophisticated approaches are better than brute-force 
solutions based on pattern matching.

\subsubsection{Indexes}\label{section:algorithms}

\paragraph{Brute force (\Brute).} These algorithms simply sort the document identifiers in the range $\DA[\ell..r]$ and report each of them once. \BruteD{} stores $\DA$ in $n \log d$ bits, while \BruteL{} retrieves the range $\SA[\ell..r]$ with the $\locateq$ functionality of the $\CSA$ and uses bitvector $B$ to convert it to $\DA[\ell..r]$. 

\paragraph{Sadakane (\Sada).}
This family of algorithms is based on the improvements of \cite{Sad07} to the algorithm of \cite{Mut02}. \SadaCL{} is the original algorithm, while \SadaCD{} uses an explicit document array $\DA$ instead of retrieving the document identifiers with $\locateq$. 

\paragraph{ILCP (\textsf{ILCP}).}
This is our proposal in Section~\ref{sec:repet}. The algorithms are the same
as those of \cite{Sad07}, but they run on the run-length encoded \ILCP\ array.
As for \Sada, \SadaIL{} obtains the document identifiers using
$\locateq$ on the $\CSA$, whereas \SadaID{} stores array $\DA$ explicitly.

\paragraph{Wavelet tree (\WT).} This index stores the document array in a
wavelet tree (Section~\ref{sec:ranksel}) to efficiently find the distinct elements in $\DA[\ell..r]$ \citep{VM07}. The best known implementation of this idea \citep{NPVjea13} uses plain, entropy-compressed, and grammar-compressed bitvectors in the wavelet tree -- depending on the level. Our \WT{} implementation uses a heuristic similar to the original WT-alpha \citep{NPVjea13}, multiplying the size of the plain bitvector by $0.81$ and the size of the entropy-compressed bitvector by $0.9$, before choosing the smallest one for each level of the tree. These constants were determined by experimental tuning.

\paragraph{Precomputed document lists (\PDL).} This is our proposal in
Section~\ref{sec:pdl-list}. Our implementation resorts to \BruteL\ to handle
the short regions that the index does not cover. The variant \PDLBC{} compresses
sets of equal documents using a Web graph compressor \citep{HNspire12.3}.
\PDLRP{} uses Re-Pair compression \citep{LM00} as implemented by Navarro\footnote{\url{http://www.dcc.uchile.cl/gnavarro/software}} and stores the dictionary in plain form. We use block size $b=256$ and storing factor $\beta=16$, which
have proved to be good general-purpose parameter values. 

\paragraph{Grammar-based (\Grammar).} This index~\citep{CM13} is an adaptation
of a grammar-compressed self-index \citep{CN12} to document listing.
Conceptually similar to \PDL, \Grammar{} uses Re-Pair to parse the collection. For each nonterminal symbol in the grammar, it stores the set of identifiers of the documents whose encoding contains the symbol. A second round of Re-Pair is used to compress the sets. Unlike most of the other solutions, \Grammar{} is an independent index and needs no $\CSA$ to operate.

\paragraph{Lempel-Ziv (\LZ).}
This index~\citep{FN13} is an adaptation of a pattern-matching index based on LZ78 
parsing \citep{Nav03} to document listing. Like \Grammar, \LZ{} does not need a $\CSA$.

\medskip

We implemented \Brute, \Sada, \ILCP, and the \PDL{} variants ourselves\footnote{\url{http://jltsiren.kapsi.fi/rlcsa}} and modified existing implementations of \WT, \Grammar, and \LZ{} for our purposes. 
We always used the RLCSA~\citep{MNSV09} as the $\CSA$, as it performs well on repetitive collections. The $\locateq$ support in RLCSA includes optimizations for long query ranges and repetitive collections, which is important for \BruteL{} and \SadaIL. We used suffix array sample periods $8, 16, 32, 64, 128$ for non-repetitive collections and $32, 64, 128, 256, 512$ for repetitive ones.

When a document listing solution uses a $\CSA$, we start the queries from the lexicographic range $[\ell..r]$ instead of the pattern $P$. This allows us to see the performance differences between the fastest solutions better. The average time required for obtaining the ranges was $4$ to $6$ microseconds per pattern, depending on the collection, which is negligible compared to the average time used by \Grammar{} (at least $760$ microseconds) and \LZ{} (at least $170$ microseconds).

\subsubsection{Results}

\paragraph{Real collections.}
Figures~\ref{figure:doclist-small} and \ref{figure:doclist-large} contain the results for document listing with small and large real collections, respectively. For most of the indexes, the time/space trade-off is given by the RLCSA sample period. The trade-off of \LZ{} comes from a parameter specific to that structure involving RMQs \citep{FN13}. \Grammar{} has no trade-off.

\BruteL{} always uses the least amount of space, but it is also the slowest solution.
In collections with many short documents (i.e., all except \Page), we have $\occ/\df < 4$
on the average. The additional effort made by \SadaCL{} and \SadaIL{} to report each
document only once does not pay off, and the space used by the RMQ structure is better
spent on increasing the number of suffix array samples for \BruteL.
The difference is, however,
very noticeable on \Page, where the documents are large and there are hundreds of
occurrences of the pattern in each document. \SadaIL{} 
uses less space than \SadaCL{} when the collection is repetitive and contains many
similar documents (i.e., on \Revision\ and \Influenza); otherwise \SadaCL{} is slightly smaller.

The two \PDL{} alternatives usually achieve similar performance, but in some
cases \PDLBC{}
uses much less space. \PDLBC, in turn, can use significantly more space than \BruteL, \SadaCL, 
and \SadaIL, but is always orders of magnitude faster. 
The document sets of versioned collections such as \Page{} and \Revision{} are
very compressible, making the collections very suitable for \PDL. On the other hand,
grammar-based compression cannot reduce the size of the stored document sets enough
when the collections are non-repetitive.
Repetitive but unstructured collections like \Influenza{} represent an interesting
special case. When the number of revisions of each base document is much larger than the
block size $b$, each leaf block stores an essentially random subset of
the revisions, which cannot be compressed very well.

Among the other indexes, \SadaCD{} and \SadaID{} can be significantly
faster than \PDLBC, but they also use much more space. From the non-$\CSA$-based
indexes, \Grammar{} reaches the Pareto-optimal curve on \Revision{} and
\Influenza{}, while being too slow or too large on the other collections. We did not build \Grammar{} for the large version of \Page, as it would have taken several months.

\begin{figure}[p]
\centering
\minipage{177pt}
  \includegraphics[trim = 0 43.2 0 0, width=\linewidth]{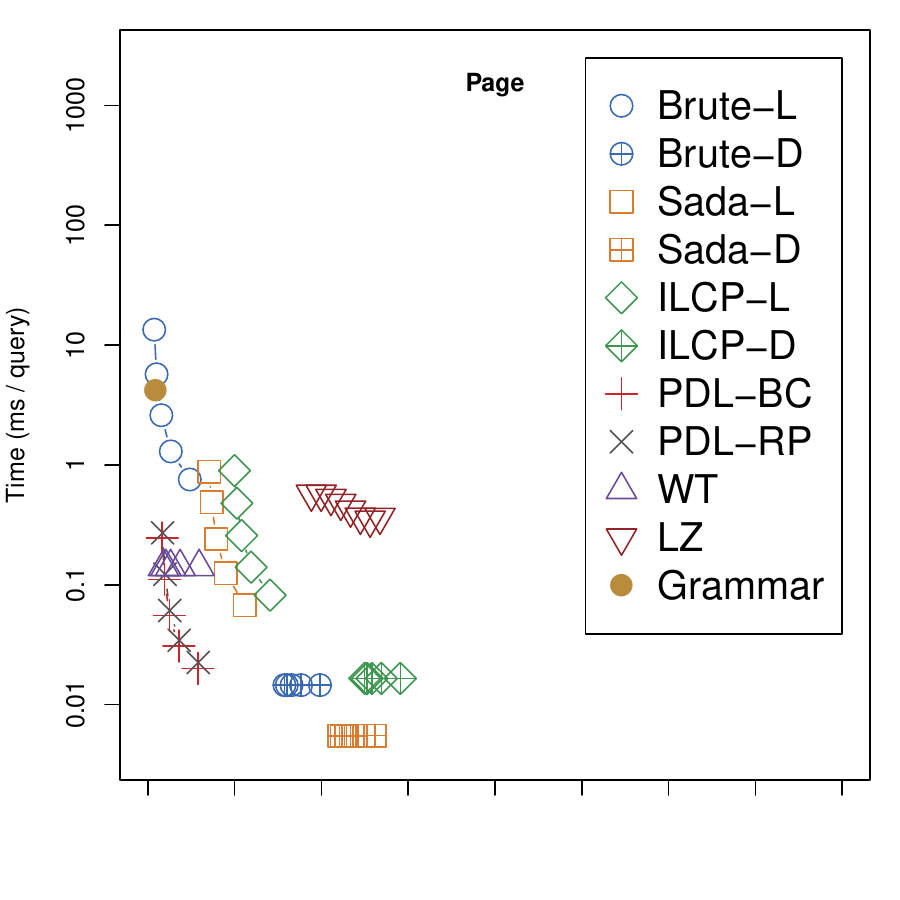}
\endminipage\hfill
\minipage{159.3pt}
  \includegraphics[trim = 43.2 43.2 0 0, width=\linewidth]{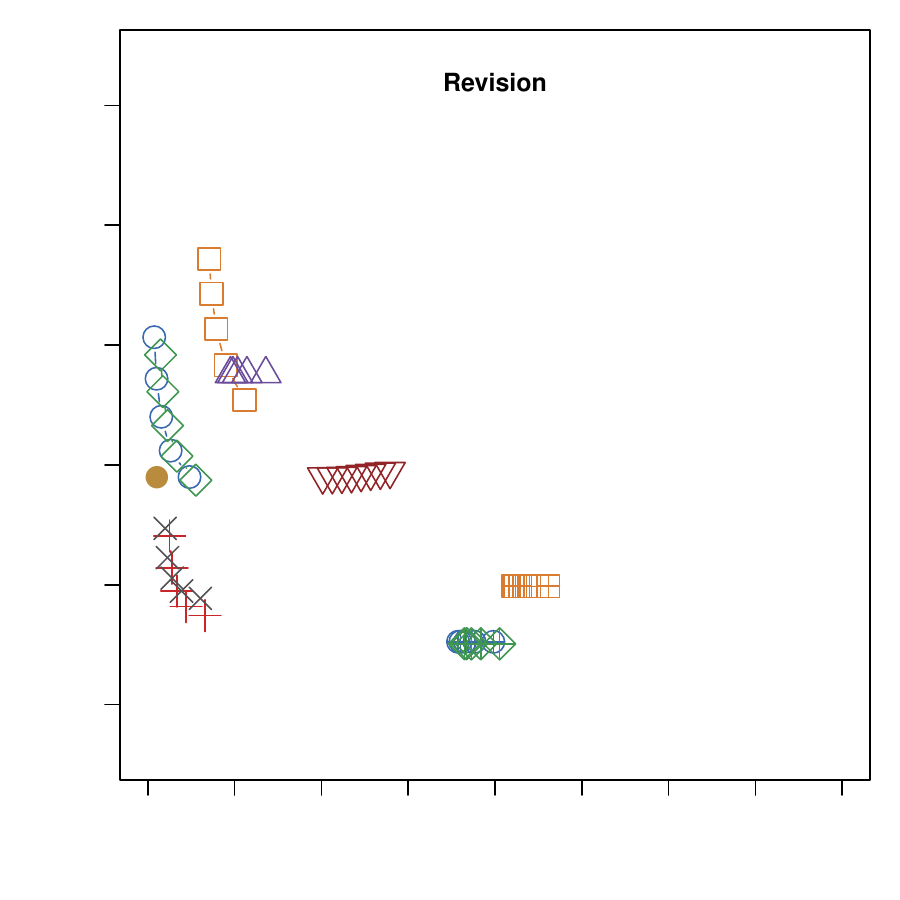}
\endminipage
\newline
\minipage{177pt}
  \includegraphics[trim = 0 43.2 0 0, width=\linewidth]{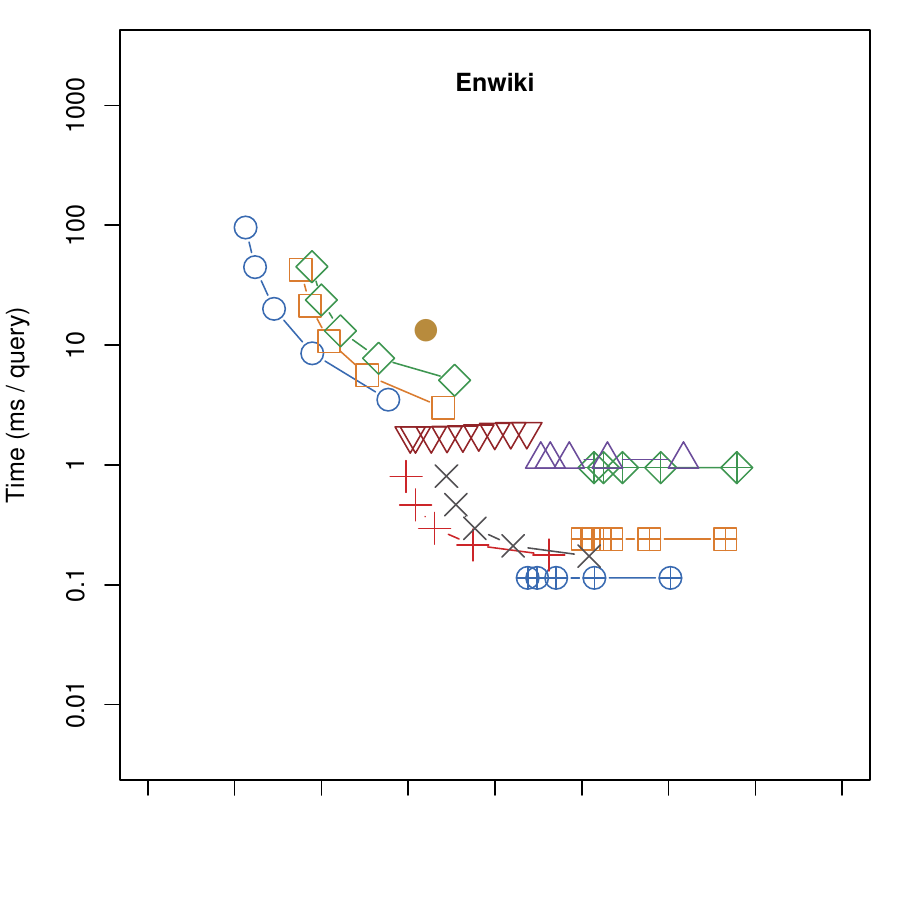}
\endminipage\hfill
\minipage{159.3pt}
  \includegraphics[trim = 43.2 43.2 0 0, width=\linewidth]{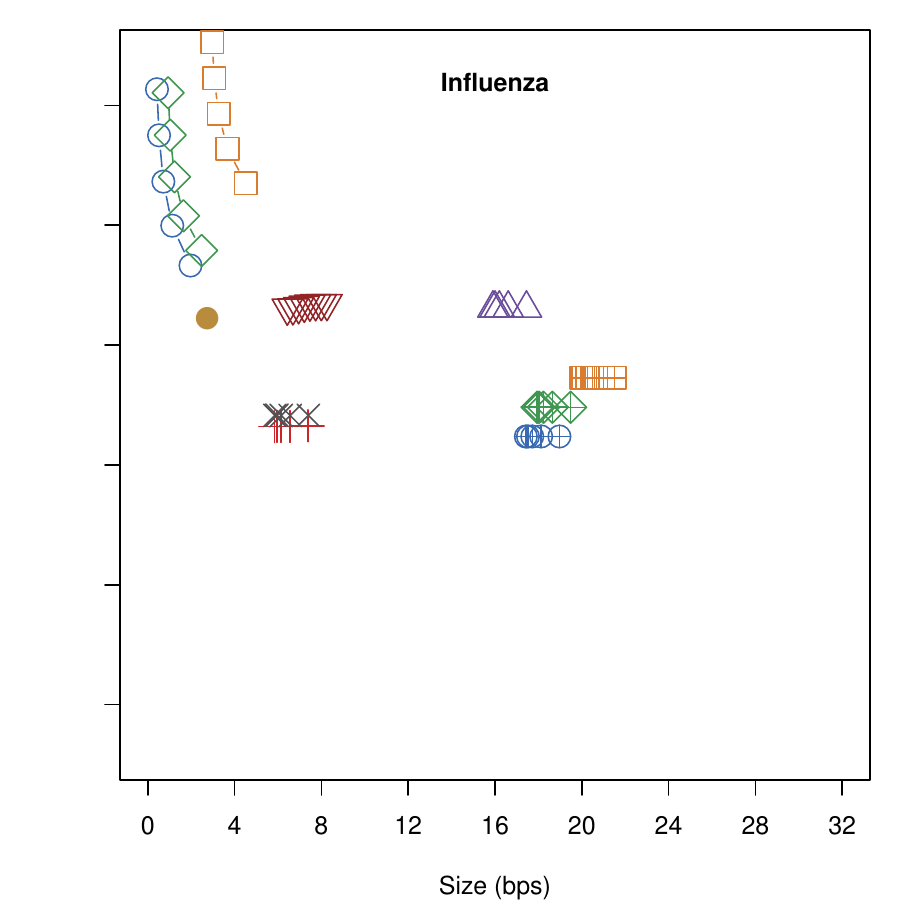}
\endminipage
\newline
\minipage{177pt}
  \includegraphics[trim = 0 0 0 0, width=\linewidth]{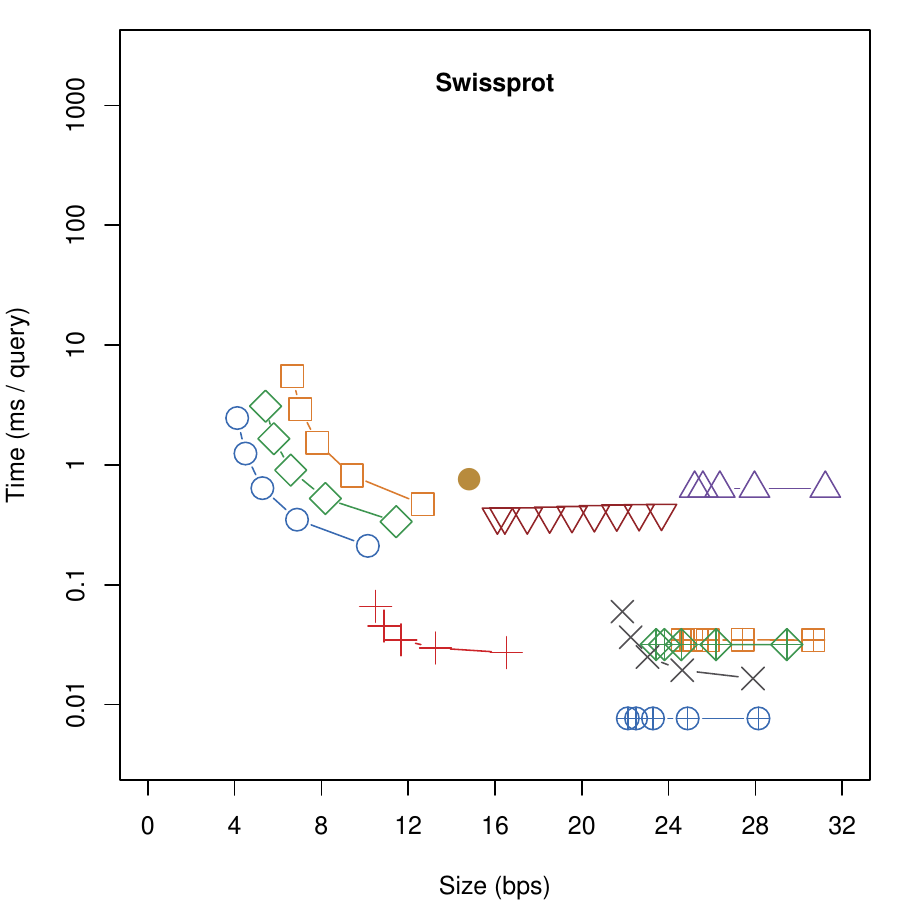}
\endminipage\hfill
\minipage{159.3pt}
\endminipage

\caption{Document listing on small real
collections. The total size of the index in bits per symbol ($x$)
and the average time per query in milliseconds ($y$).}
\label{figure:doclist-small}
\end{figure}

\begin{figure}[t]
\centering
\minipage{177pt}
  \includegraphics[trim = 0 43.2 0 0, width=\linewidth]{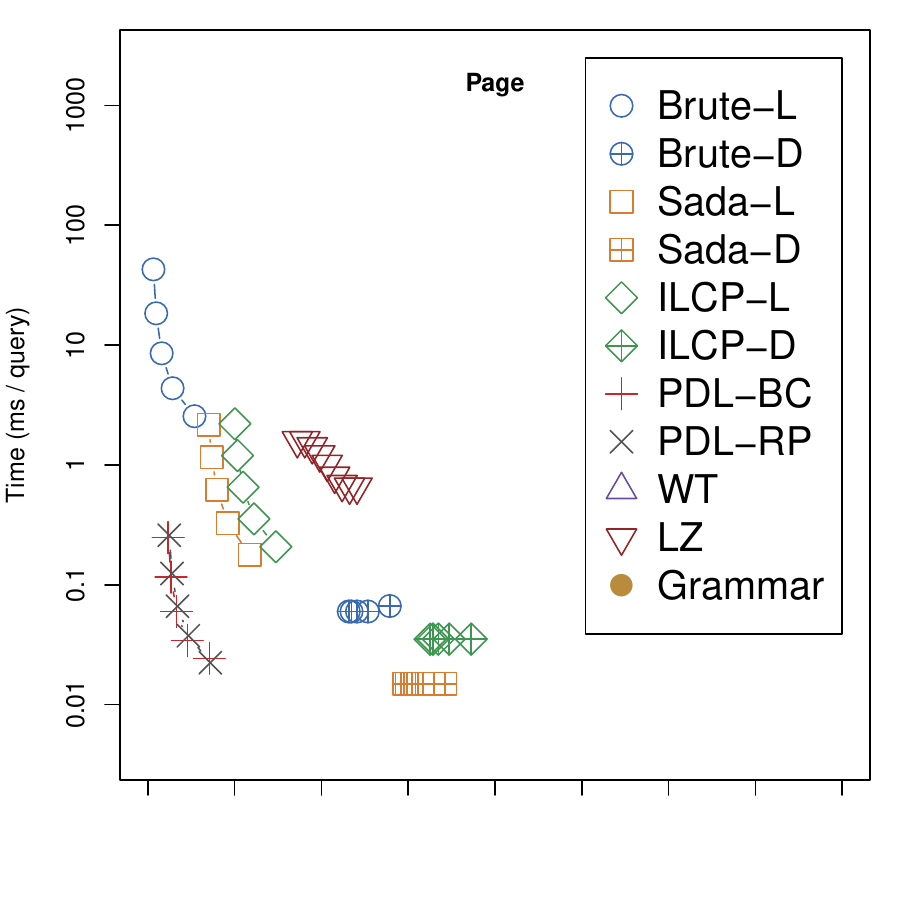}
\endminipage\hfill
\minipage{159.3pt}
  \includegraphics[trim = 43.2 43.2 0 0, width=\linewidth]{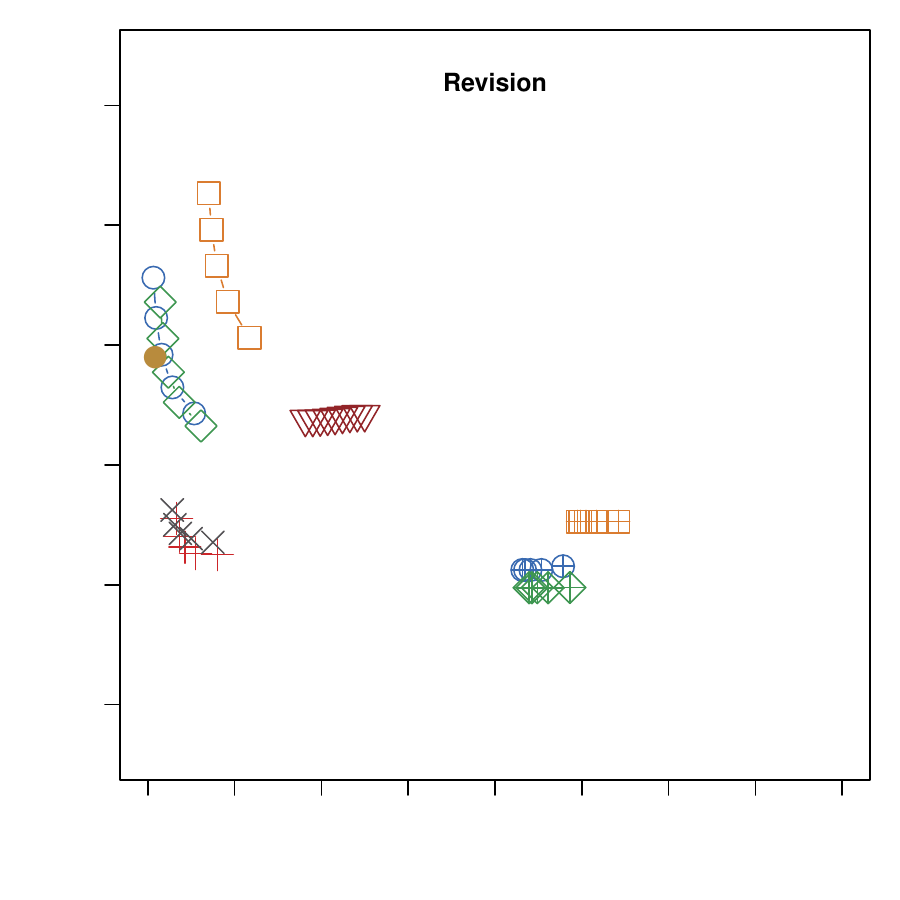}
\endminipage
\newline
\minipage{177pt}
  \includegraphics[trim = 0 0 0 0, width=\linewidth]{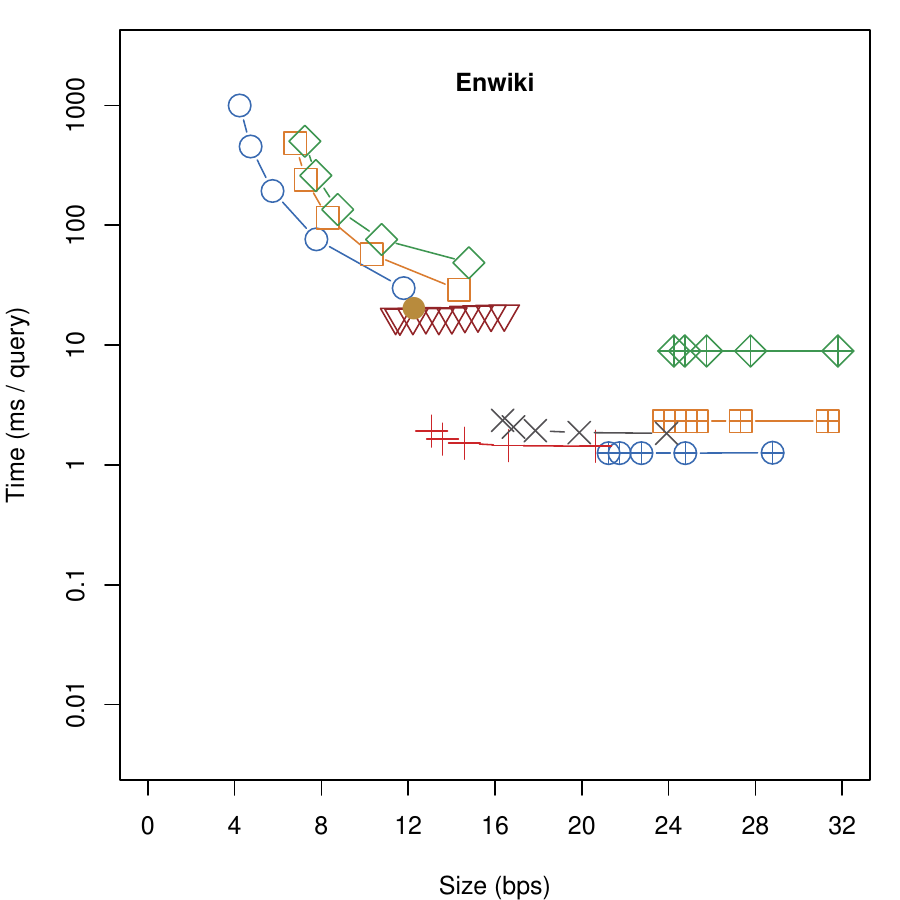}
\endminipage\hfill
\minipage{159.3pt}
  \includegraphics[trim = 43.2 0 0 0, width=\linewidth]{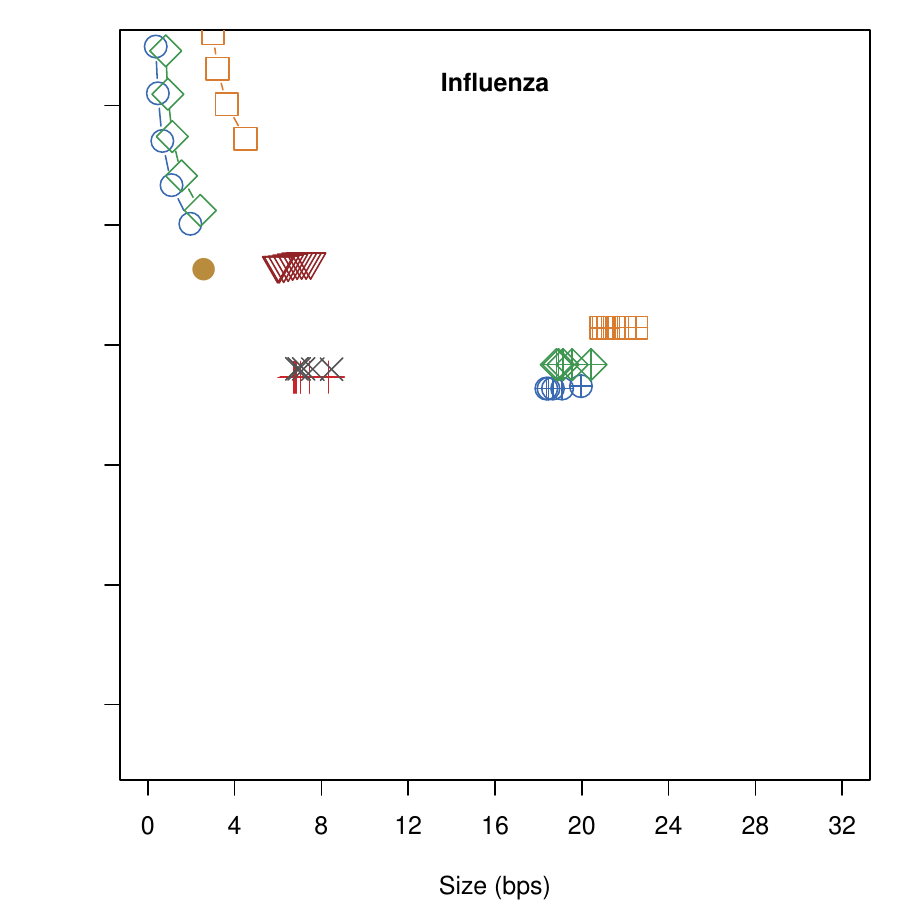}
\endminipage

\caption{Document listing on large real
collections. The total size of the index in bits per symbol ($x$)
and the average time per query in milliseconds ($y$).}
\label{figure:doclist-large}
\end{figure}

In general, we can recommend \PDLBC{} as a medium-space alternative for document
listing. When less space is available, we can use \SadaIL, which offers robust
time and space guarantees. If the documents are small, we can even use
\BruteL. Further,
we can use fast document counting to compare $\df$ with $\occ =
r-\ell+1$, and choose between \SadaIL{} and \BruteL{} according to the results.

\paragraph{Synthetic collections.}

\begin{figure}[t!]
\minipage{177pt}
  \includegraphics[trim = 0 43.2 0 0, width=\linewidth]{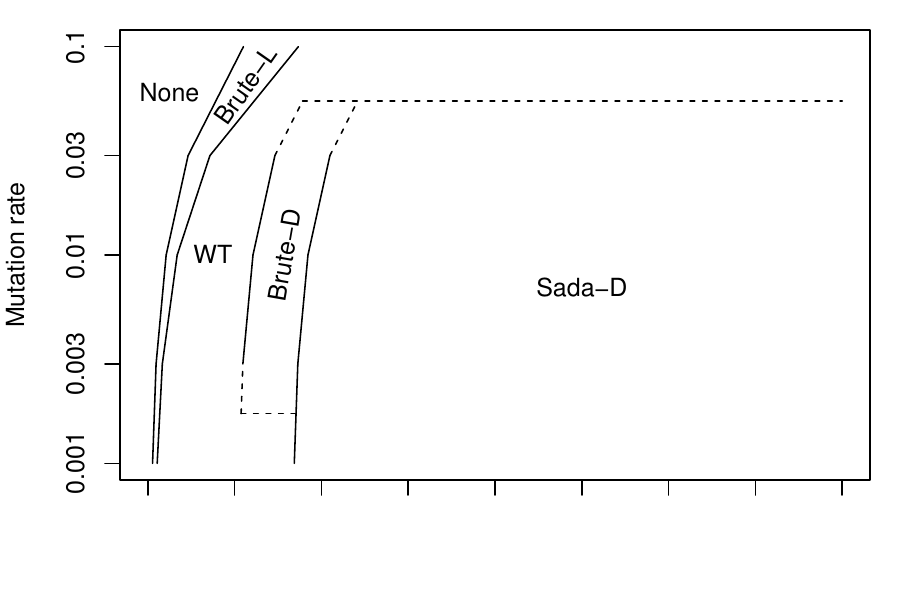}
\endminipage\hfill
\minipage{159.3pt}
  \includegraphics[trim = 43.2 43.2 0 0, width=\linewidth]{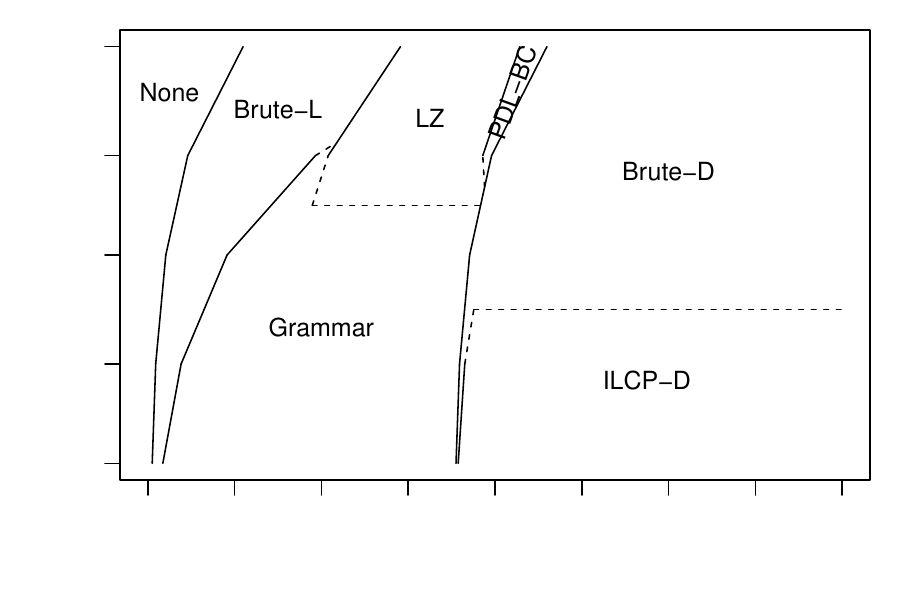}
\endminipage
\newline
\minipage{177pt}
  \includegraphics[trim = 0 43.2 0 0, width=\linewidth]{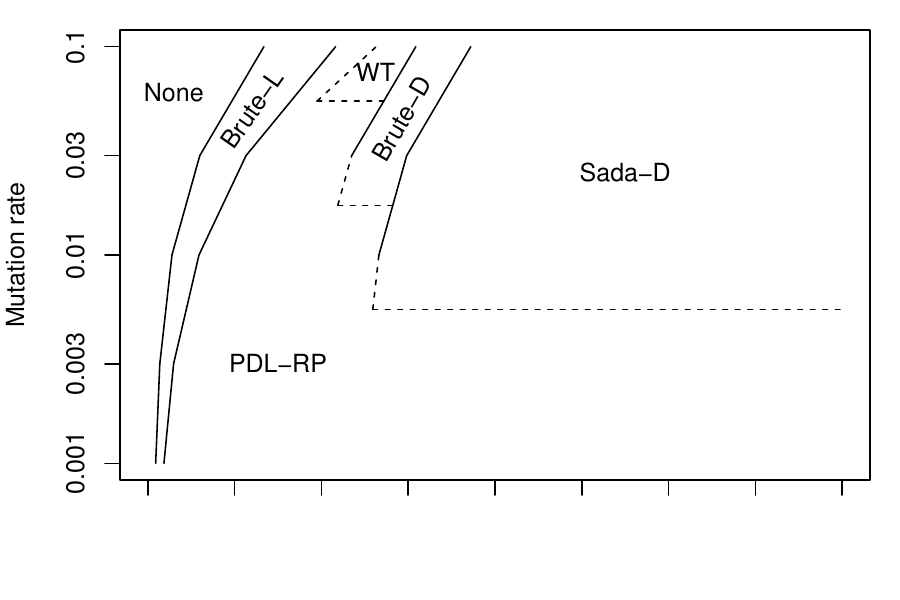}
\endminipage\hfill
\minipage{159.3pt}
  \includegraphics[trim = 43.2 43.2 0 0, width=\linewidth]{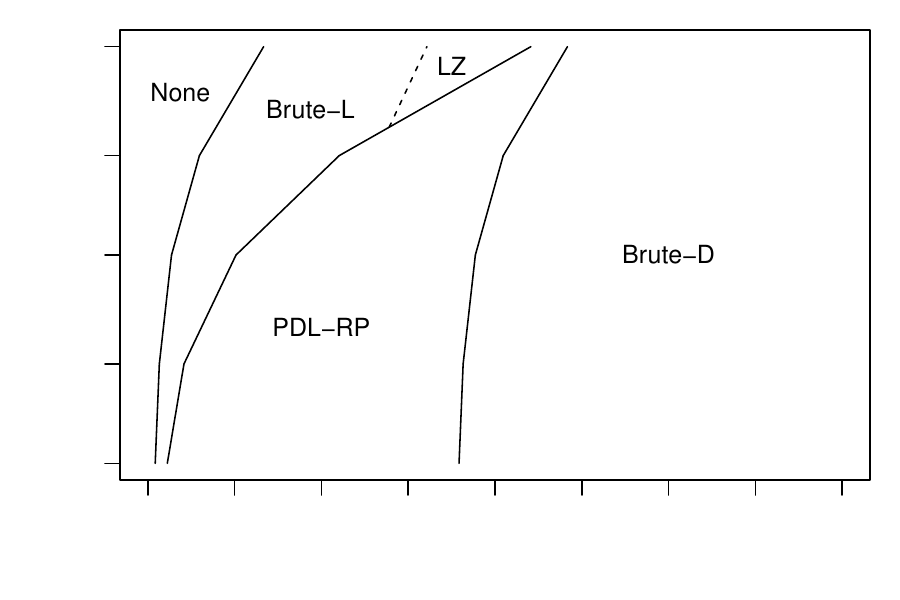}
\endminipage
\newline
\minipage{177pt}
  \includegraphics[trim = 0 0 0 0, width=\linewidth]{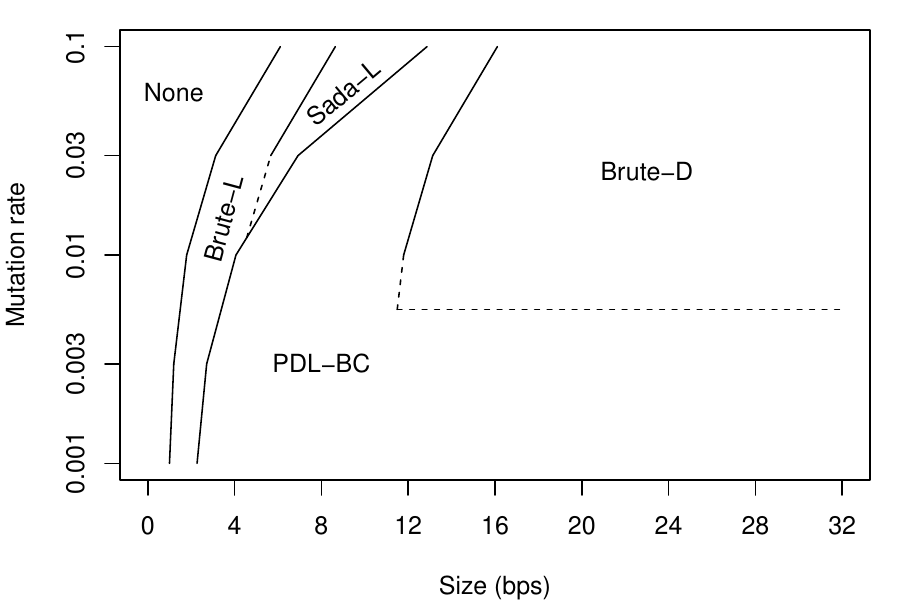}
\endminipage\hfill
\minipage{159.3pt}
  \includegraphics[trim = 43.2 0 0 0, width=\linewidth]{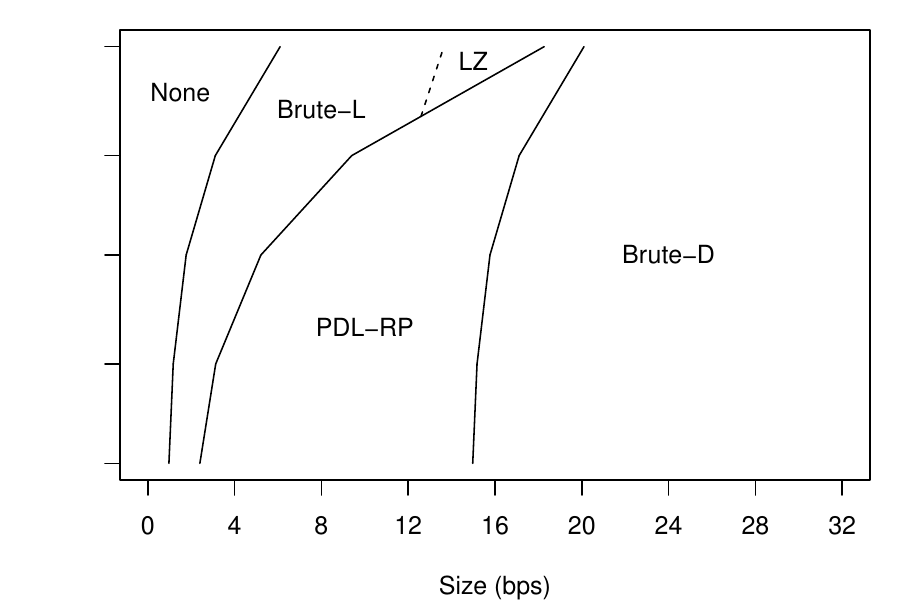}
\endminipage

\caption{Document listing on synthetic collections. The fastest solution for a given size in bits per symbol and a mutation rate. From top to bottom: $10$, $100$, and $1{,}000$ base documents with \Concat{} (left) and \Version{} (right). \textsf{None} denotes that no solution can achieve that size.}
\label{figure:synthetic-wiki}
\end{figure}

\begin{figure}[t!]
\minipage{177pt}
  \includegraphics[trim = 0 43.2 0 0, width=\linewidth]{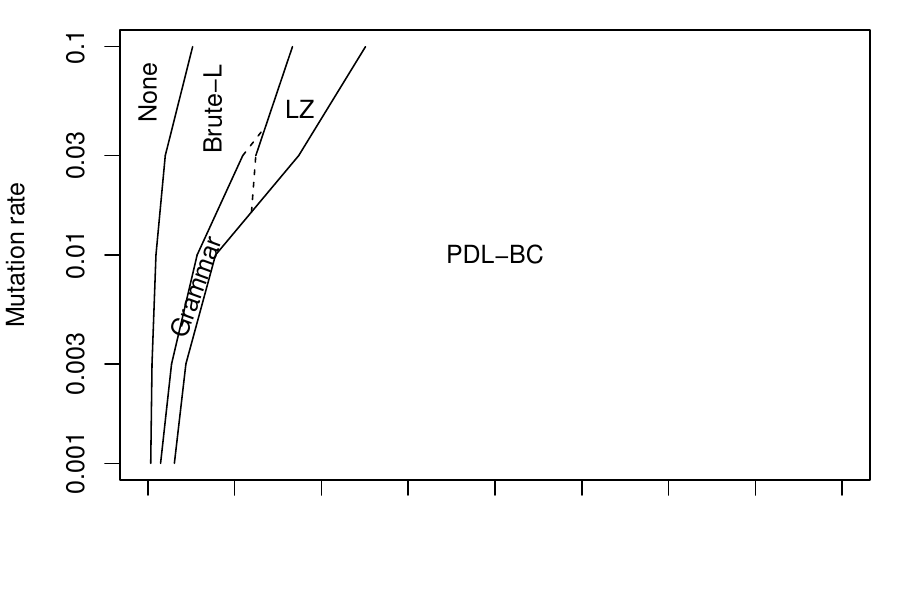}
\endminipage\hfill
\minipage{159.3pt}
  \includegraphics[trim = 43.2 43.2 0 0, width=\linewidth]{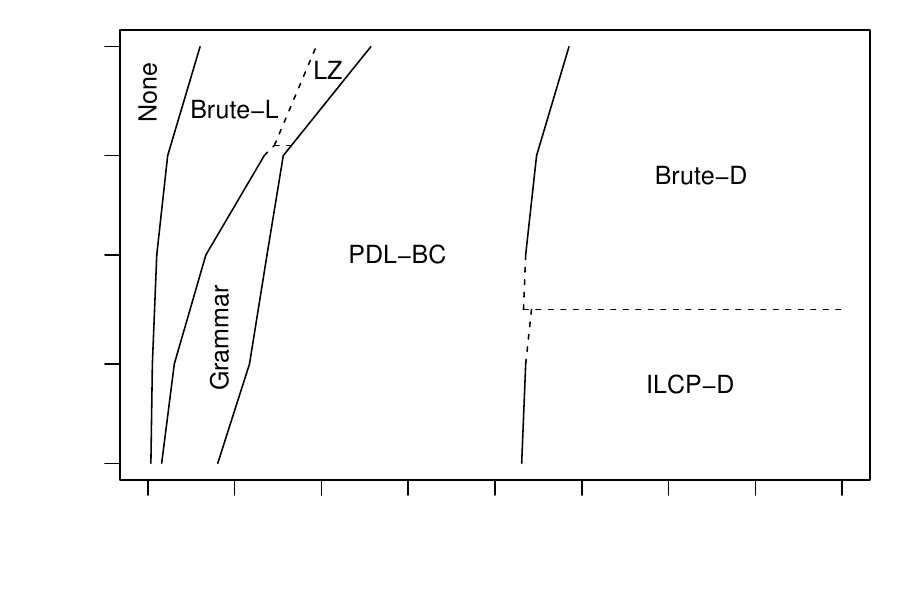}
\endminipage
\newline
\minipage{177pt}
  \includegraphics[trim = 0 0 0 0, width=\linewidth]{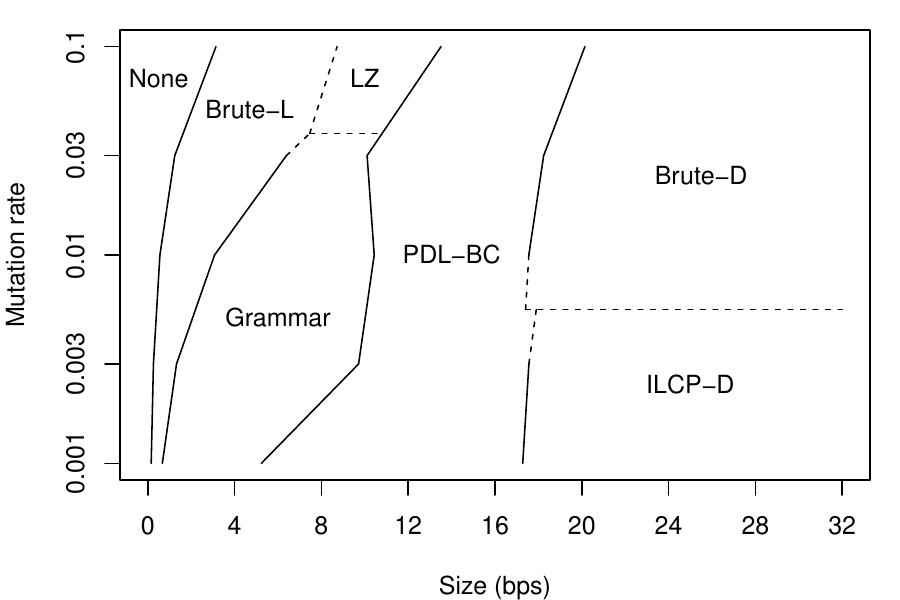}
\endminipage\hfill
\minipage{159.3pt}
  \includegraphics[trim = 43.2 0 0 0, width=\linewidth]{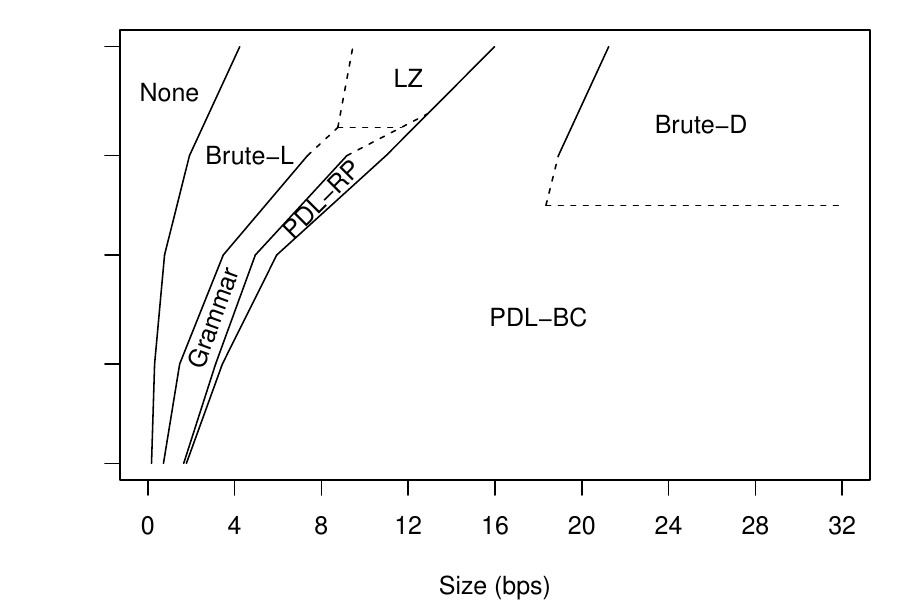}
\endminipage

\caption{Document listing on synthetic collections. The fastest solution for a given size in bits per symbol and a mutation rate. \DNA{} with $1$ (top left), $10$ (top right), $100$ (bottom left), and $1{,}000$ (bottom right) base documents. \textsf{None} denotes that no solution can achieve that size.}
\label{figure:synthetic-dna}
\end{figure}

Figures~\ref{figure:synthetic-wiki} and \ref{figure:synthetic-dna} show our document listing results with synthetic collections. Due to the large number of collections, the results for a given collection type and number of base documents are combined in a single plot, showing the fastest algorithm for a given amount of space and mutation rate. Solid lines connect measurements that are the fastest for their size, while dashed lines are rough interpolations.

The plots were simplified in two ways. Algorithms providing a marginal
and/or inconsistent improvement in speed in a very narrow region (mainly
\SadaCL{} and \SadaIL) were left out. When \PDLBC{} and \PDLRP{} had a very
similar performance, only one of them was chosen for the plot.

On \DNA, \Grammar{} was a good solution for small mutation rates, while \LZ{} was good with larger mutation rates. With more space available, \PDLBC{} became the fastest algorithm. \BruteD{} and \SadaID{} were often slightly faster than \PDL, when there was enough space available to store the document array. On \Concat{} and \Version{}, \PDL{} was usually a good mid-range solution, with \PDLRP{} being usually smaller than \PDLBC{}. The exceptions were the collections with $10$ base documents, where the number of variants ($1{,}000$) was clearly larger than the block size ($256$). With no other structure in the collection, \PDL{} was unable to find a good grammar to compress the sets. At the large end of the size scale, algorithms using an explicit document array $\DA$ were usually the fastest choices.

\subsection{Top-$k$ Retrieval}\label{sec:topk-experiments}

\subsubsection{Indexes}

We compare the following top-$k$ retrieval algorithms. Many of them share names
with the corresponding document listing structures described in
Section~\ref{section:algorithms}.

\paragraph{Brute force (\Brute).} These algorithms correspond to the document listing 
algorithms \BruteD\ and \BruteL. To perform top-$k$ retrieval, we not only
collect the distinct document identifiers after sorting $\DA[\ell..r]$, we also
record the number of times each one appears. The $k$ identifiers appearing
most frequently are then reported.

\paragraph{Precomputed document lists (\PDL).} We use the variant of \PDLRP\ modified for top-$k$ retrieval, as described in Section~\ref{sec:topk}. \PDLtopk{$b$} denotes PDL with block size $b$ and with document sets for all suffix tree nodes above the leaf blocks, while \PDLtopk{$b$+F} is the same with term frequencies. \PDLtopk{$b$--$\beta$} is PDL with block size $b$ and storing factor $\beta$.

\paragraph{Large and fast (\SURF).}
This index~\citep{Gog2015a} is based on a conceptual idea by \cite{NN12}, and improves upon a previous implementation \citep{KN13}. It can answer top-$k$ queries quickly if the pattern occurs at least twice in each reported document. If documents with just one occurrence are needed, \SURF{} uses a variant of \SadaCL{} to find them.

\medskip

We implemented the \Brute{} and \PDL{} variants ourselves\footnote{\url{http://jltsiren.kapsi.fi/rlcsa}} and used the existing implementation of \SURF{}\footnote{\url{https://github.com/simongog/surf/tree/single_term}}.
While \WT{} \citep{NPVjea13} also supports top-$k$ queries, the 32-bit implementation cannot index the large versions of the document collections used in the experiments. As with document listing, we subtracted the time required for finding the lexicographic ranges $[\ell..r]$ using a $\CSA$ from the measured query times. \SURF{} uses a $\CSA$ from the SDSL library \citep{Gog2014b}, while the rest of the indexes use RLCSA.

\subsubsection{Results}

Figure~\ref{figure:single-term} contains the results for top-$k$ retrieval using the large versions of the real collections. We left \Page{} out of the results, as the number of documents ($280$) was too low for meaningful top-$k$ queries. For most of the indexes, the time/space trade-off is given by the RLCSA sample period, while the results for \SURF{} are for the three variants presented in the paper.

\begin{figure}[p]
\minipage{177pt}
  \includegraphics[trim = 0 43.2 0 0, width=\linewidth]{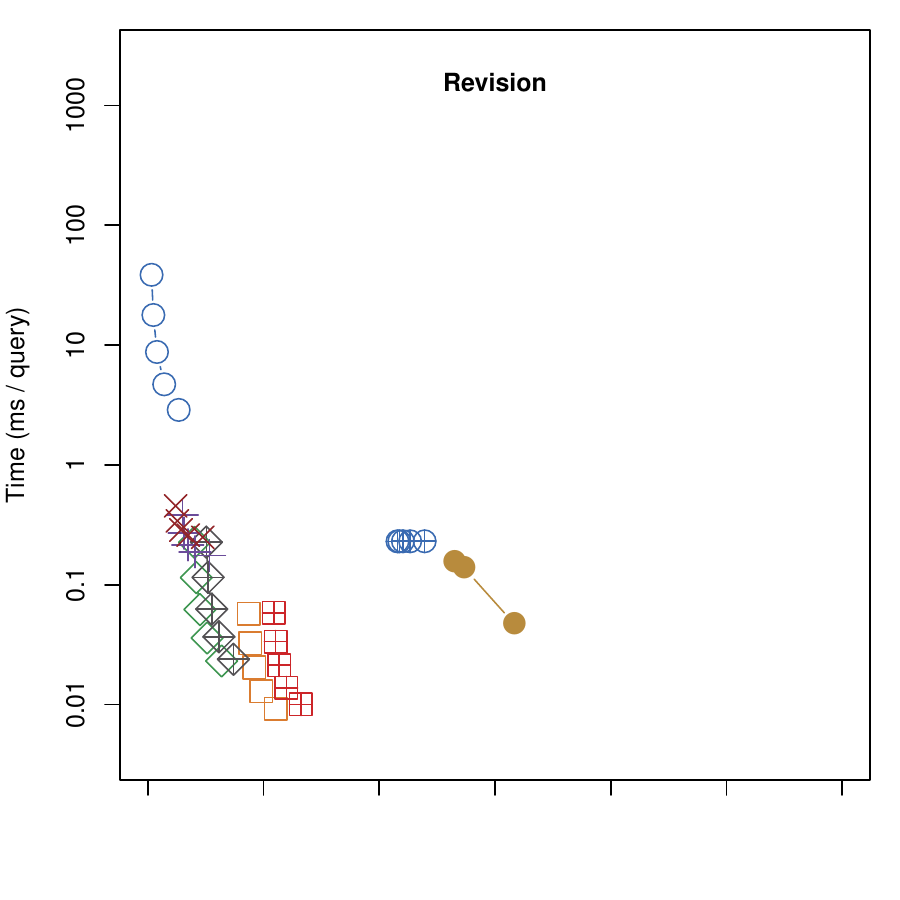}
\endminipage\hfill
\minipage{159.3pt}
  \includegraphics[trim = 43.2 43.2 0 0, width=\linewidth]{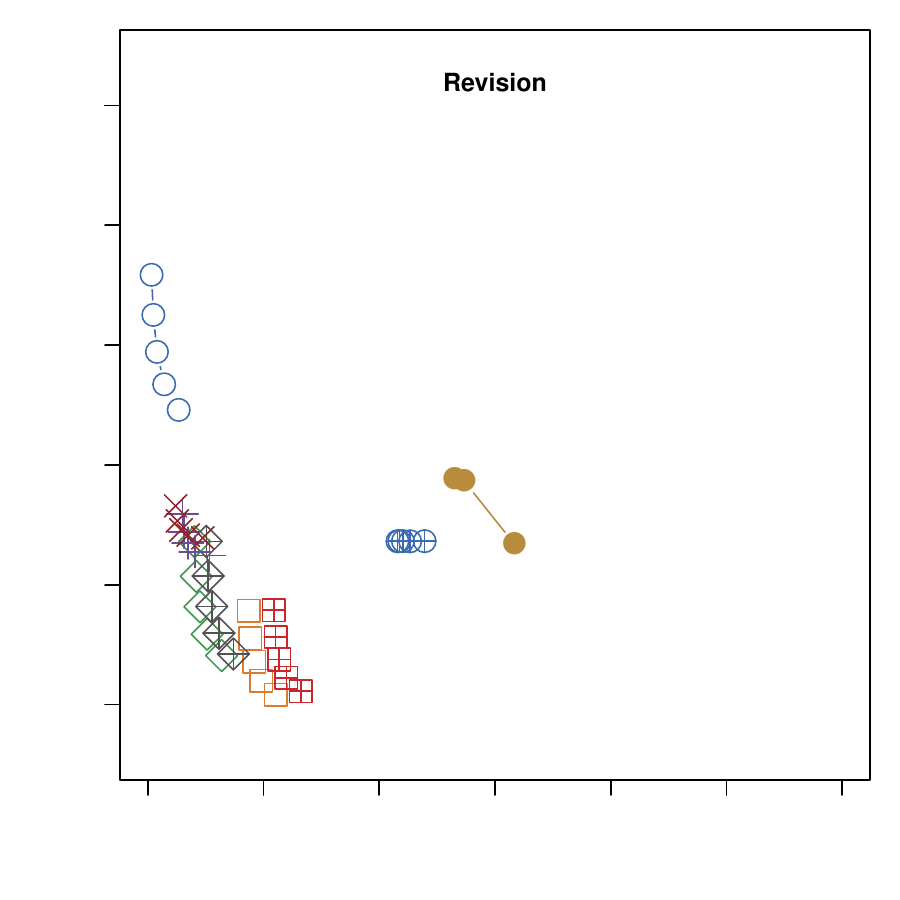}
\endminipage
\newline
\minipage{177pt}
  \includegraphics[trim = 0 43.2 0 0, width=\linewidth]{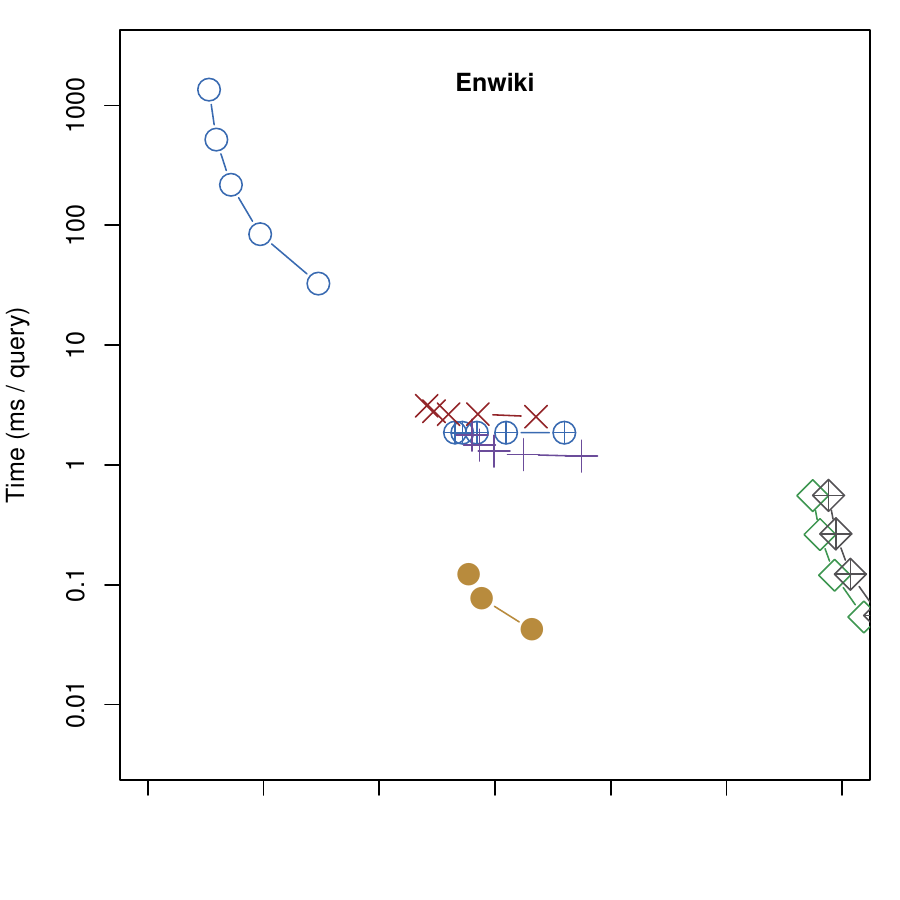}
\endminipage\hfill
\minipage{159.3pt}
  \includegraphics[trim = 43.2 43.2 0 0, width=\linewidth]{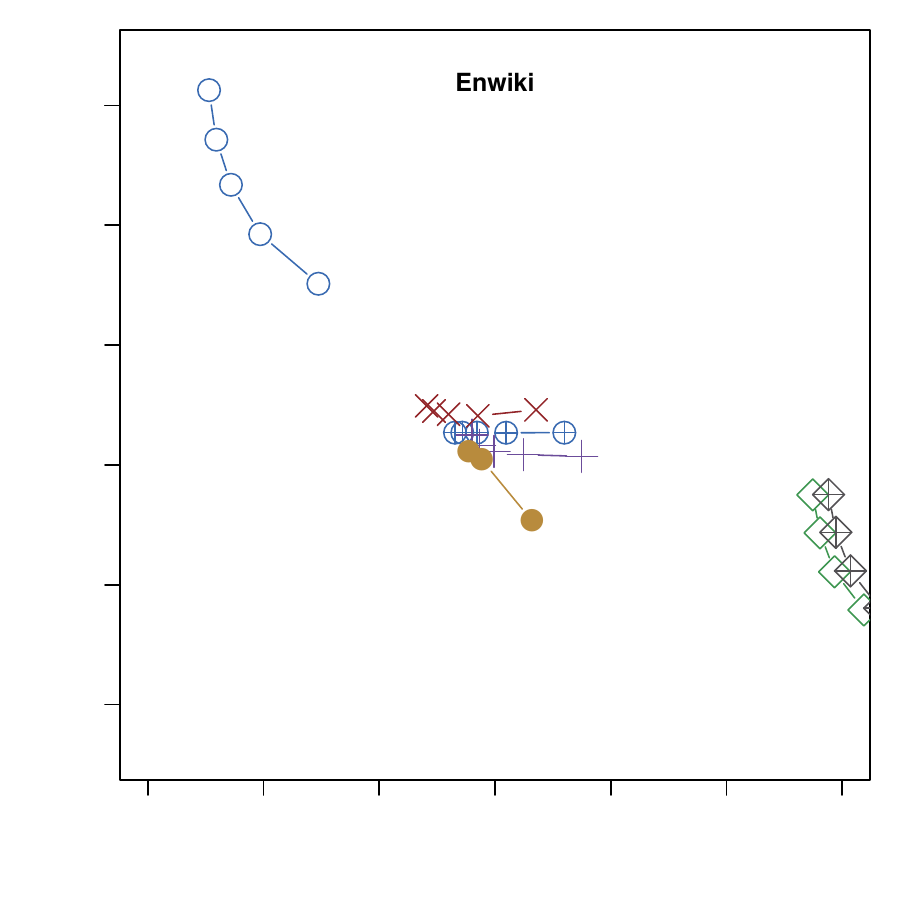}
\endminipage
\newline
\minipage{177pt}
  \includegraphics[trim = 0 0 0 0, width=\linewidth]{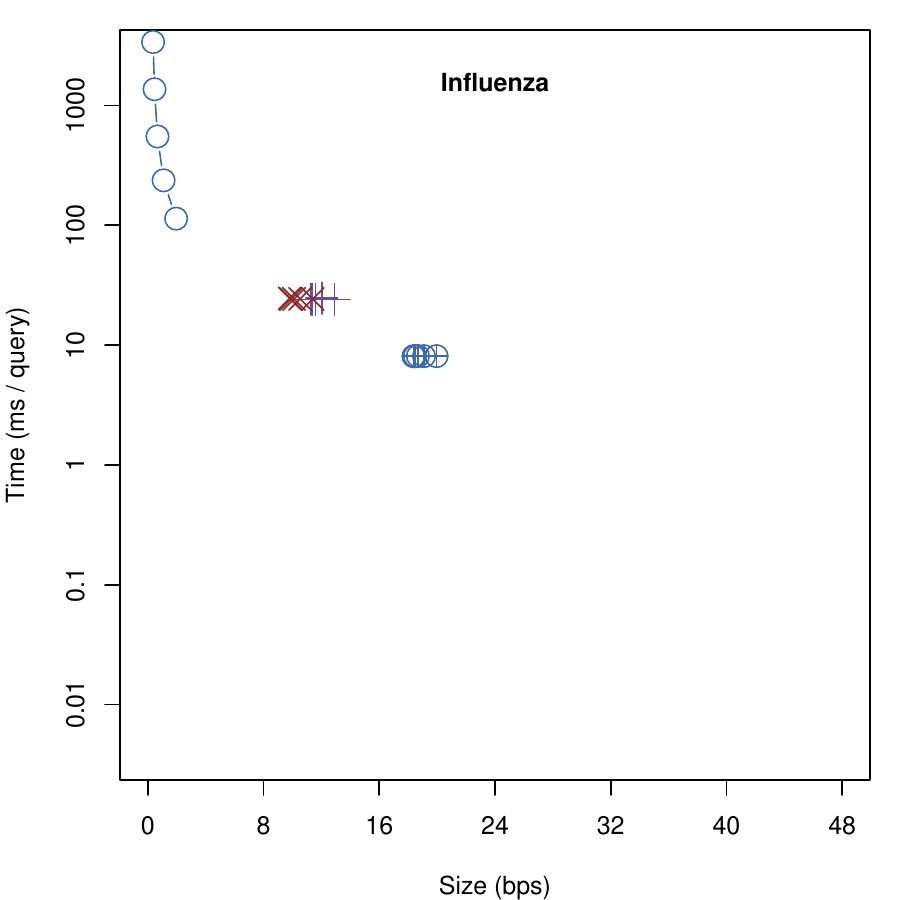}
\endminipage\hfill
\minipage{159.3pt}
  \includegraphics[trim = 43.2 0 0 0, width=\linewidth]{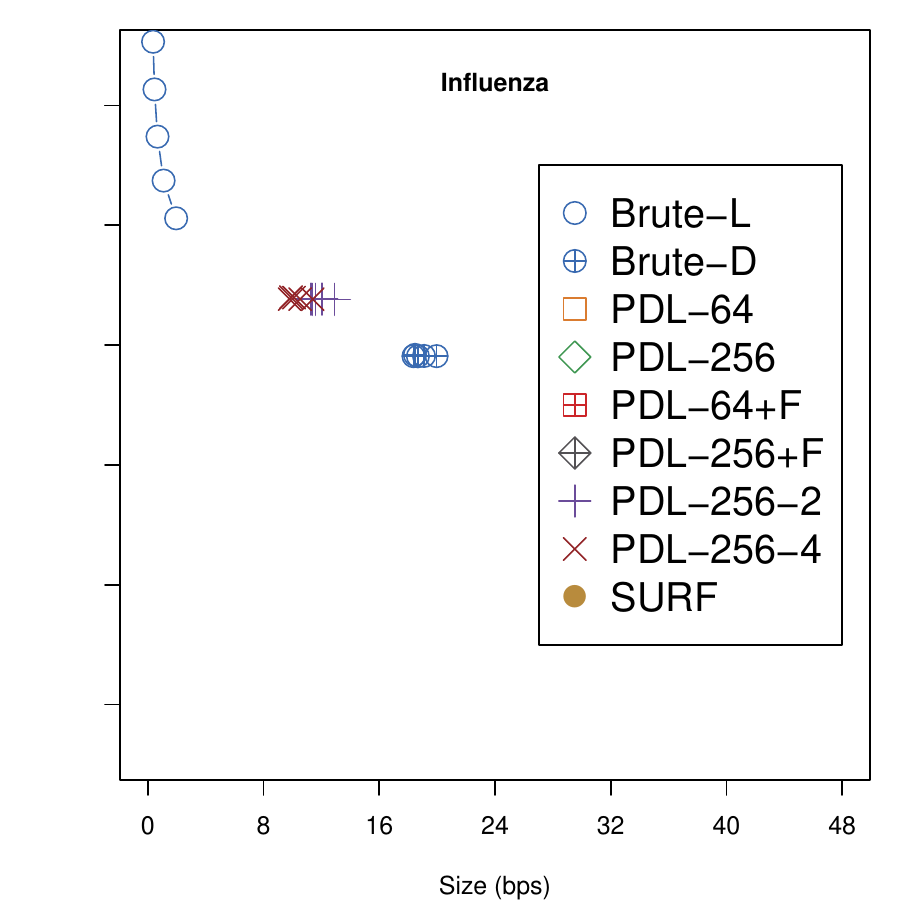}
\endminipage

\caption{Single-term top-$k$ retrieval on real collections with $k = 10$ (left) and $k = 100$ (right). The total size of the index in bits per symbol (x) and the average time per query in milliseconds (y).}
\label{figure:single-term}
\end{figure}

The three collections proved to be very different. With \Revision, the \PDL{} variants were both fast and space-efficient. When storing factor $\beta$ was not set, the total query times were dominated by rare patterns, for which \PDL{} had to resort to using \BruteL. This also made block size $b$ an important time/space trade-off. When the storing factor was set, the index became smaller and slower and the trade-offs became less significant. \SURF{} was larger and faster than \BruteD{} with $k = 10$ but became slow with $k=100$.

On \Enwiki{}, the variants of \PDL{} with storing factor $\beta$ set had a performance similar to \BruteD{}. \SURF{} was faster with roughly the same space usage. \PDL{} with no storing factor was much larger than the other solutions. However, its time performance became competitive for $k = 100$, as it was almost unaffected by the number of documents requested.

The third collection, \Influenza, was the most surprising of the three. \PDL{} with storing factor $\beta$ set was between \BruteL{} and \BruteD{} in both time and space. We could not build \PDL{} without the storing factor, as the document sets were too large for the Re-Pair compressor. The construction of \SURF{} 
also failed with this dataset.

\subsection{Document Counting}\label{sec:doccount-experiments}

\subsubsection{Indexes}

We use two fast document listing algorithms as baseline document counting methods (see Section~\ref{section:algorithms}): \BruteD{} sorts the query range 
$\DA[\ell..r]$ to count the number of distinct document identifiers, and \PDLRP{} 
returns the length of the list of documents obtained. Both indexes use the
RLCSA with suffix array sample period set to $32$ on non-repetitive datasets,
and to $128$ on repetitive datasets.

We also consider a number of encodings of Sadakane's document counting structure (see Section~\ref{sec:count}). The following ones encode the bitvector $H'$ directly in a number of ways:

\begin{itemize}

\item \Sada{} uses a plain bitvector representation.

\item \SadaR{} uses a run-length encoded bitvector as supplied in the RLCSA implementation.
It uses $\delta$-codes to represent run lengths and packs them into blocks of 32 bytes of
encoded data. Each block stores how many bits and 1s are there before it.

\item \sadaR{} uses a run-length encoded bitvector, represented with a sparse
bitmap \citep{OS07} marking the beginnings of the 0-runs and another for the 1-runs.

\item \sadaD{} uses run-length encoding with $\delta$-codes to represent the
lengths. Each block in the bitvector contains the encoding of 128 1-bits, while
three sparse bitmaps are used to mark the number of bits,
1-bits, and starting positions of block encodings.

\item \SadaG{} uses a grammar-compressed bitvector \citep{NO14}.

\end{itemize}

The following encodings use filters in addition to bitvector $H'$:

\begin{itemize}

\item \SadaPG{} uses \Sada{} for $H'$ and a gap-encoded bitvector for the filter
bitvector $F$. The gap-encoded bitvector is also provided in the RLCSA implementation.
It differs from the run-length encoded bitvector by only encoding runs of 0-bits.

\item \SadaPR{} uses \Sada{} for $H'$ and \SadaR{} for $F$.

\item \SadaRG{} uses \SadaR{} for $H'$ and a gap-encoded bitvector for $F$.

\item \SadaRR{} uses \SadaR{} for both $H'$ and $F$.

\item \sadaSS{} uses sparse bitmaps for both $H'$ and the sparse filter $F_{S}$.

\item \sadaS{} is \sadaSS{} with an additional sparse bitmap for the 1-filter $F_{1}$

\item \sadaRS{} uses \sadaR{} for $H'$ and a sparse bitmap for $F_{1}$.

\item \sadaDS{} uses \sadaD{} for $H'$ and a sparse bitmap for $F_{1}$.

\end{itemize}

Finally, \wt{} implements the technique described in
Section~\ref{sec:ilcp-count},
using the same encoding as in \sadaR{} to represent the bitvectors in the wavelet tree.

Our implementations of the above methods can be found online.\footnote{\url{http://jltsiren.kapsi.fi/rlcsa} and \url{https://github.com/ahartik/succinct}}

\subsubsection{Results}

Due to the use of 32-bit variables in some of the implementations, we could not build all structures for the large real collections. Hence we used the medium versions of \Page, \Revision, and \Enwiki{}, the large version of \Influenza, and the only version of \Swissprot{} for the benchmarks. We started the queries from precomputed lexicographic ranges $[\ell..r]$ in order to emphasize the differences between the fastest variants. For the same reason, we also left out of the plots the size of the RLCSA and the possible document retrieval structures. Finally, as it was almost always the fastest method, we scaled the plots to leave out anything much larger than plain \Sada{}. The results can be seen in Figure~\ref{figure:counting}. Table~\ref{table:counting-results} in Appendix~\ref{appendix:results} lists the results in further detail.

\begin{figure}[p]
\includegraphics[trim = 0 43.2 0 0, width=\textwidth]{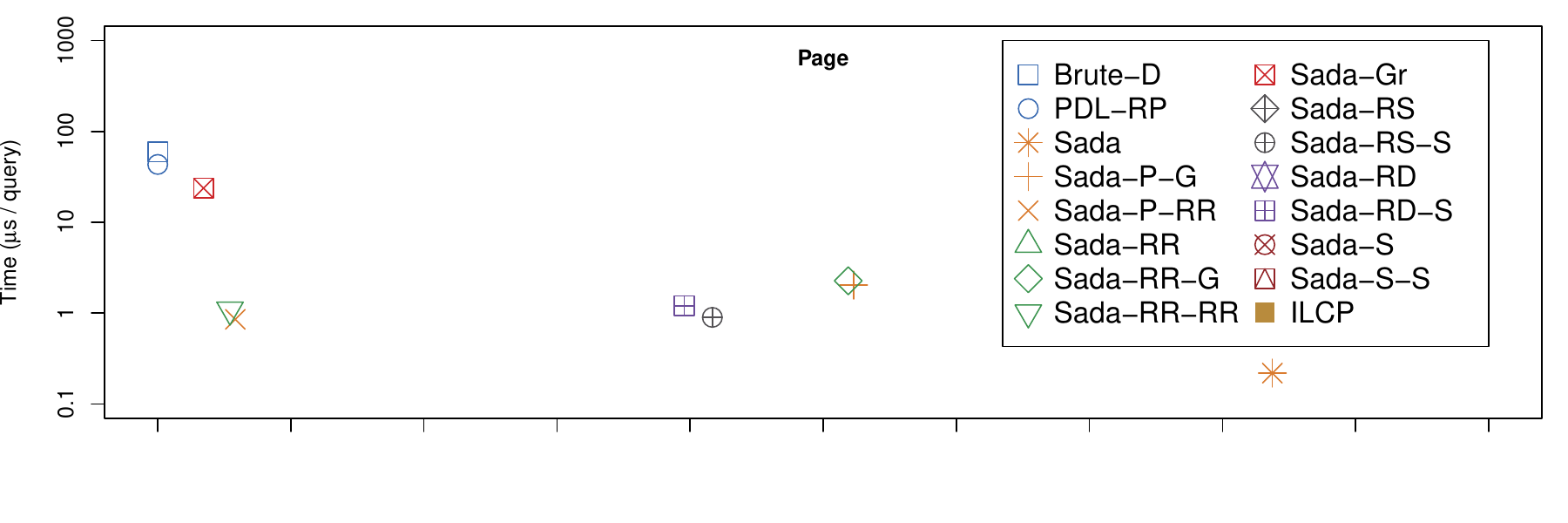}
\includegraphics[trim = 0 43.2 0 0, width=\textwidth]{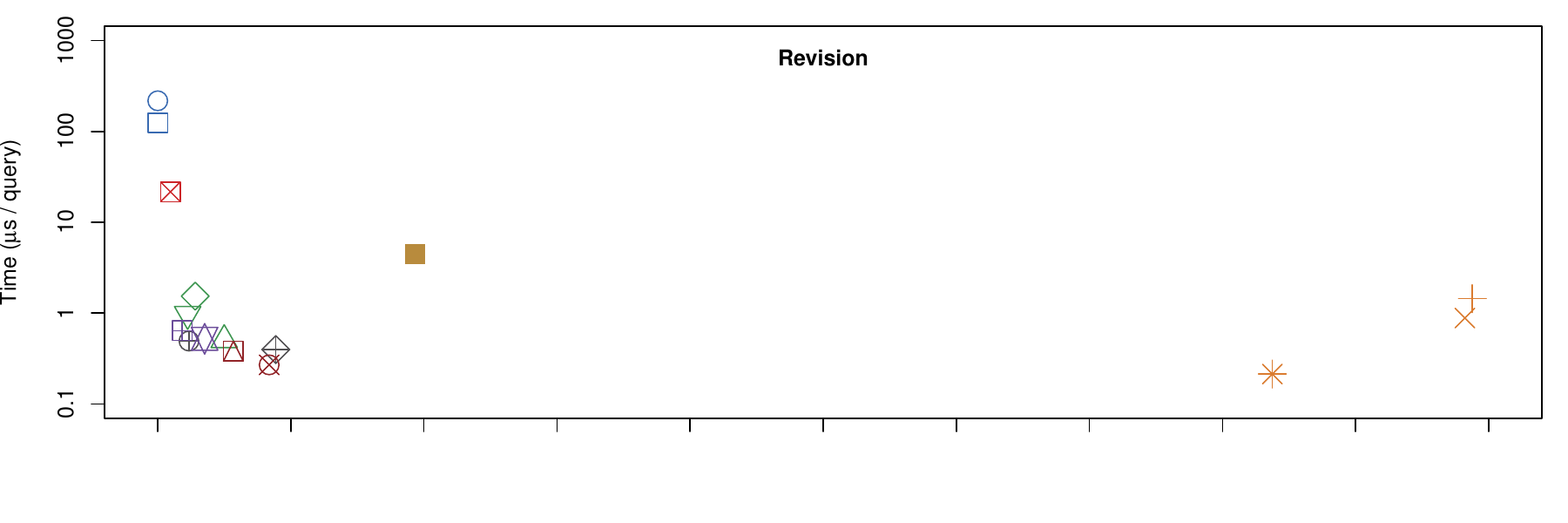}
\includegraphics[trim = 0 43.2 0 0, width=\textwidth]{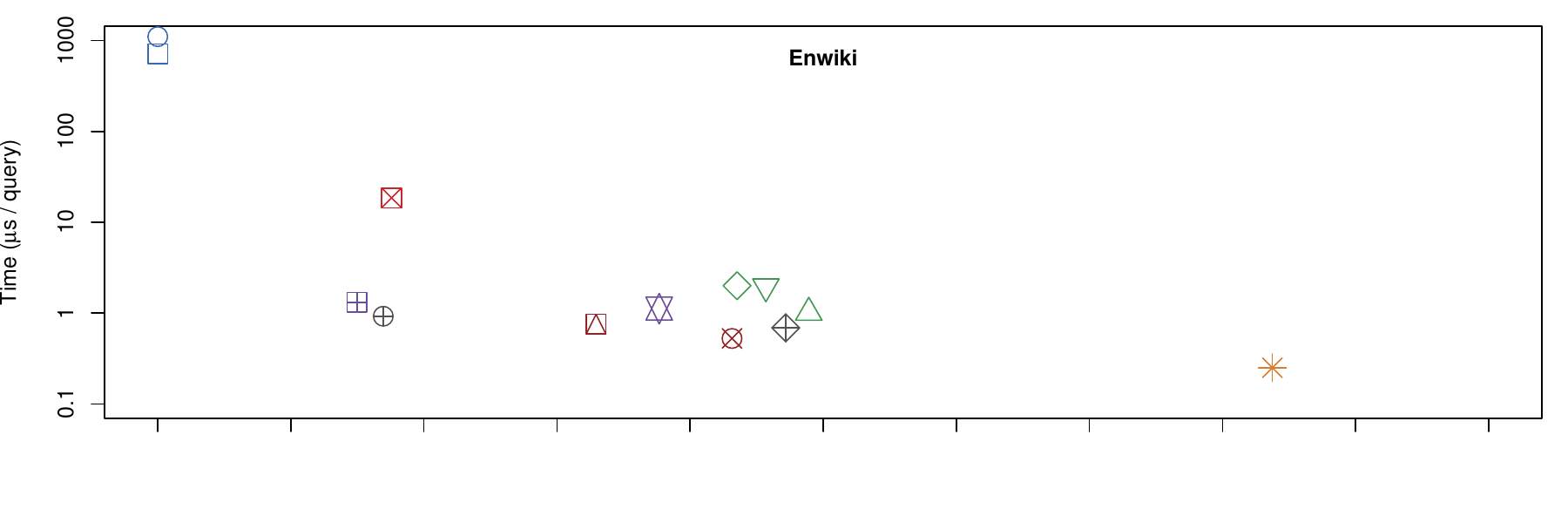}
\includegraphics[trim = 0 43.2 0 0, width=\textwidth]{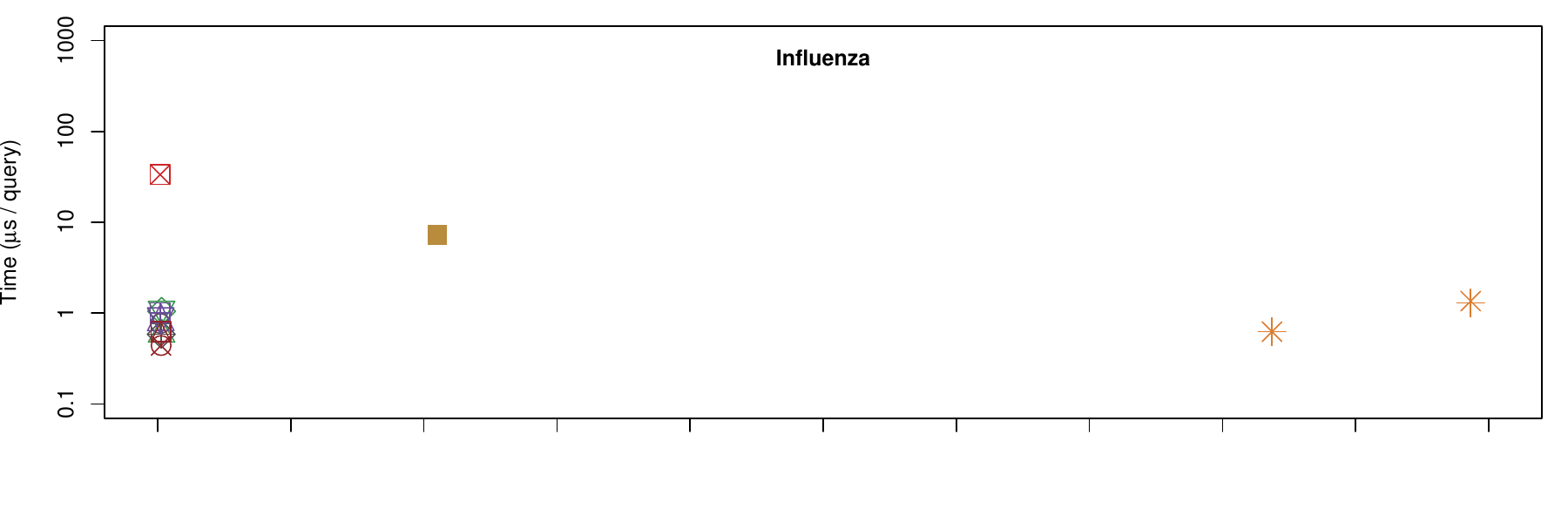}
\includegraphics[width=\textwidth]{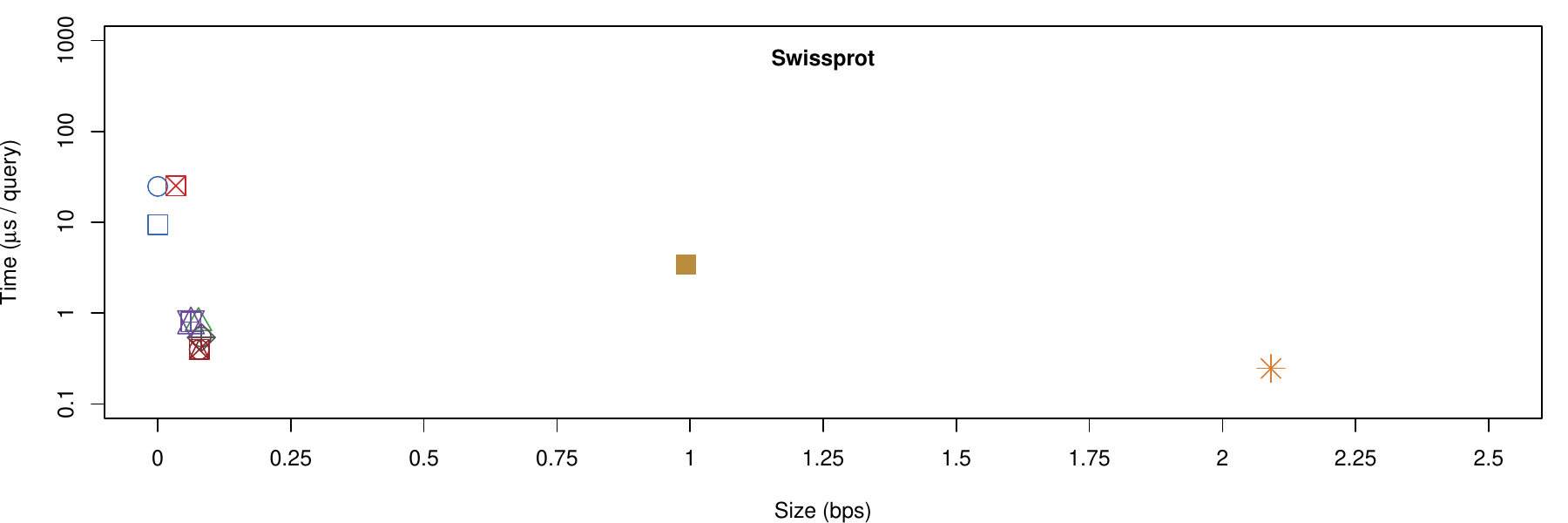}

\caption{Document counting on different datasets. The size of the counting structure in bits per symbol (x) and the average query time in microseconds (y). The baseline document listing methods are presented as having size $0$, as they take advantage of the existing functionalities in the index.}
\label{figure:counting}
\end{figure}

On \Page{}, the filtered methods \SadaPR{} and \SadaRR{} are clearly the best choices, being only slightly larger than the baselines and orders of magnitude faster. Plain \Sada{} is much faster than those, but it takes much more space than all the other indexes.
Only \SadaG{} compresses the structure better, but it is almost as slow
as the baselines.

On \Revision{}, there were many small encodings with similar performance. Among those, \sadaRS{} is the fastest. \sadaSS{} is somewhat larger and faster. As on \Page{}, plain \Sada{} is even faster, but it takes much more space.

The situation changes on the non-repetitive \Enwiki. Only \sadaDS, \sadaRS,
and \SadaG{} can compress the bitvector clearly below $1$~bit per symbol, and \SadaG{} is much slower than the other two. At around $1$~bit per symbol, \sadaSS{} is again the fastest option. Plain \Sada{} requires twice as much space as \sadaSS, but is also twice as fast.

\Influenza{} and \Swissprot{} contain, respectively, RNA and protein sequences, making
each individual document quite random.  Such collections are easy cases for Sadakane's
method, and many encodings compress the bitvector very well. In both cases, \sadaSS{}
was the fastest small encoding. On \Influenza, the small encodings fit in CPU cache,
making them often faster than plain \Sada.

Different compression techniques succeed with different collections, for different reasons, which
complicates a simple recommendation for a best option. Plain \Sada{} is always fast,
while \sadaSS{} is usually smaller without sacrificing too much performance. When more
space-efficient solutions are required, the right choice depends on the type of the collection.
Our ILCP-based structure, \wt, also outperforms \Sada{} in space on most collections, but it
is always significantly larger and slower than compressed variants of \Sada. 

\subsection{The Multi-term \tfidf\ Index}

We implement our multi-term index as follows.
We use RLCSA as the $\CSA$, \PDLtopk{256+F} for single-term top-$k$ retrieval, and \sadaSS{} for document counting. We could have integrated the document counts into the \PDL{} structure, but a separate counting structure makes the index more flexible. Additionally, encoding the number of redundant documents in each internal node of the suffix tree (\Sada) often takes less space than encoding the total number of documents in each node of the sampled suffix tree (\PDL).
We use the basic \tfidf\ scoring scheme.

We tested the resulting performance on the 1432~MB \Wiki{} collection. RLCSA took 0.73~bps with sample period $128$ (the sample period did not have a significant impact on query performance), \PDLtopk{256+F} took 3.37~bps, and \sadaSS{} took 0.13~bps, for a total of 4.23~bps (757~MB). Out of the total of 100{,}000 queries in the query set, there were matches for 31{,}417 conjunctive queries and 97{,}774 disjunctive queries.

The results can be seen in Table~\ref{table:multi-term}. When using a single
query thread, the index can process 136--229 queries per second (around 4--7
milliseconds per query), depending on the query type and the value of $k$.
Disjunctive queries are faster than conjunctive queries, while larger values
of $k$ do not increase query times significantly. Note that our ranked
disjunctive query algorithm preempts the processing of the lists
of the patterns, whereas in the conjunctive ones we are forced to expand the
full document lists for all the patterns; this is why the former are faster.
The speedup from using 32 threads is around 18x.

\begin{table}[t!]\centering
\caption{Ranked multi-term queries on the \Wiki{} collection. Query type, number of documents requested, and the average number of queries per second with 1, 8, 16, and 32 query threads.}\label{table:multi-term}

\begin{tabular}{cccccc}
\hline
\noalign{\smallskip}
Query      & $k$ & 1 thread & 8 threads & 16 threads & 32 threads \\
\noalign{\smallskip}
\hline
\noalign{\smallskip}
Ranked-AND &  10 &      152 &       914 &       1699 & 2668 \\
           & 100 &      136 &       862 &       1523 & 2401 \\
\noalign{\smallskip}
Ranked-OR  &  10 &      229 &      1529 &       2734 & 4179 \\
           & 100 &      163 &      1089 &       1905 & 2919 \\
\noalign{\smallskip}
\hline
\end{tabular}
\end{table}

Since our multi-term index offers a functionality similar to basic inverted
index queries, it seems sensible
to compare it to an inverted index designed for natural language texts. For this purpose, we indexed the \Wiki{} collection using Terrier~\citep{Macdonald2012} version~4.1 with the default settings. See Table~\ref{table:terrier} for a comparison between the two indexes. 

Note that the similarity in the functionality is only superficial: our index
can find {\em any text substring}, whereas the inverted index can only look
for indexed {\em words and phrases}. Thus our index has an index point per
symbol, whereas Terrier has an index point per word (in addition, inverted indexes usually discard words deemed uninteresting, like stopwords). Note that \PDL\ also chooses frequent strings and builds their lists of documents, but since it has many more index points, its posting lists are 200 times longer than those of Terrier, and the number of lists is 300 times larger. Thanks to the compression of its lists, however, \PDL\ uses only 8 times more space than Terrier. On the other hand, both indexes have similar query performance. When logging and output was set to minimum, Terrier could process 231~top-10 queries and 228~top-100 queries per second under the \tfidf\ scoring model using a single query thread.

\begin{table}[t!]\centering
\caption{Our index (PDL) and an inverted index (Terrier) on the \Wiki{} collection. The size of the vocabulary, the posting lists, and the collection in millions of elements, the size of the index in megabytes, and the number of Ranked-OR queries per second with $k = 10$ or $100$ using a single thread.}\label{table:terrier}

\begin{tabular}{ccccccc}
\hline
\noalign{\smallskip}
Index   & Vocabulary & Posting lists & Collection & Size & \multicolumn{2}{c}{Queries / second} \\
\noalign{\smallskip}
\hline
\noalign{\smallskip}
PDL     & 39.2M      & 8840M         & 1500M      & 757  & 229      & 163 \\
        & substrings & documents     & symbols    & MB   & ($k=10$) & ($k=100$) \\
\noalign{\smallskip}
Terrier & 0.134M     & 42.3M         & 133M       & 90.1 & 231    & 228 \\
        & tokens     & documents     & tokens     & MB   & ($k=10$) & ($k=100$) \\
\noalign{\smallskip}
\hline
\end{tabular}
\end{table}

\section{Conclusions}
\label{sec:concl}

We have investigated the space/time tradeoffs involved in indexing highly
repetitive string collections, with the goal of performing information retrieval
tasks on them. Particularly, we considered the problems of document listing,
top-$k$ retrieval, and document counting. We have developed new indexes that
perform particularly well on those types of collections, and studied how other
existing data structures perform in this scenario, and in which cases the 
indexes are actually better than brute-force approaches. As a result, we 
offered recommendations on which structures to use depending on the kind of 
repetitiveness involved and the desired space usage. As a proof of
concept, we have shown how the tools we developed can be assembled to build
an efficient index supporting ranked multi-term queries on repetitive string 
collections.

We do not aim to outperform inverted indexes on natural language text 
collecions, where they are unbeatable, but rather to offer similar capabilities on
generic string collections, where inverted indexes cannot be applied.
Our developments are at the level of algorithmic ideas and prototypes. In
order to have our most promising structures scale up to real-world information 
systems, where inverted indexes are now the norm, various research problems 
must be faced:
\begin{enumerate}
\item Our construction algorithms scale up to a few gigabytes. This limits
the collection sizes we can handle, even if they are repetitive and thus the
final structures are much smaller. For example, our PDL structure first 
builds the classical suffix tree and then samples it. Using
construction space proportional to that of the final structures in the case
of repetitive scenarios, or building efficiently using the disk, is an 
important research problem.
\item When the datasets are sufficiently large, even the compressed structures
will have to operate on disk. Inverted indexes are extremely disk-friendly,
which makes them perform well on huge text collections. We have not yet studied
this aspect of our structures, although PDL seems well-suited to this case: it
traverses one or a few contiguous lists (which should be decompressed in main memory)
or a contiguous area of the suffix array.
\item Our data structures are static, that is, they must be rebuilt from
scratch when documents are inserted in the collection or deleted from it.
Inverted indexes tolerate updates much better, though they are not fully
dynamic either. Instead, since in many scenarios updates are not so frequent,
popular solutions combine a large part of the collection that is indexed and
a small recent part that is traversed sequentially. It is likely that our
structures will also perform well under such a scheme, as long as we manage to
rebuild the index periodically within controlled space and time.
\item We showed that our structures can handle multi-term queries under the
simple tf-idf scoring scheme. While this can be acceptable in some applications
for generic string collections, information retrieval on natural language texts
uses, nowadays, much more sophisticated formulas. Inverted indexes have been
adapted to successfully support those formulas that are used for a first 
filtration step, such as BM25. Studying how to extend our indexes to handle these 
is another interesting research problem.
\item One point where our indexes could outperform inverted indexes is in phrase
queries, where inverted indexes must perform costly list intersections. Our
suffix-array based indexes, instead, need not do anything special. For a 
fair comparison, we should regard the text as a sequence of tokens (i.e., the
terms that are indexed by the inverted index) and build our indexes on them.
The resulting structure would then only answer term and phrase queries, just 
like an inverted index, but would be must faster at phrases.
\end{enumerate}

\begin{acknowledgements}

This work was supported in part by 
Academy of Finland grants 268324, 258308, 250345 (CoECGR), and 134287; 
the Helsinki Doctoral Programme in Computer Science; 
the Jenny and Antti Wihuri Foundation, Finland; 
the Wellcome Trust grant 098051, UK;
Fondecyt grant 1-140796, Chile; 
the Millennium Nucleus for Information and Coordination in Networks (ICM/FIC
P10-024F), Chile; 
Basal Funds FB0001, Conicyt, Chile; and
European Union’s Horizon 2020 research and innovation programme
under the Marie Sklodowska-Curie grant agreement No 690941.
Finally, we thank the reviewers for their useful comments, which helped improve
the presentation, and Meg Gagie for correcting our grammar.
\end{acknowledgements}

\clearpage
\appendix
\section{Detailed Results}\label{appendix:results}

Table~\ref{table:counting-results} shows the precise numerical results
displayed in Figure~\ref{figure:counting}, to allow for a finer-grained 
comparison.

\begin{table}[h!]\centering
\caption{Document counting on different datasets. The average query time in microseconds and the size of the counting structure in bits per symbol. Results on the Pareto frontier have been \textbf{highlighted}. The baseline document listing methods \BruteD{} and \PDLRP{} are presented as having size $0$, as they take advantage of the existing functionalities in the index. We did not build \SadaPG, \SadaPR, \SadaRG, and \SadaRR{} for \Swissprot, because the filter was empty and the remaining structure was equivalent to \Sada{} or \SadaR.}\label{table:counting-results}

\newlength{\uslength}
\settowidth{\uslength}{\textmu{}s}
\newcommand{\busw}{\makebox[\uslength][l]{b}}
\begin{tabular}{rrrrrr}
\hline
\noalign{\smallskip}
 & \multicolumn{1}{c}{\Page} & \multicolumn{1}{c}{\Revision} & \multicolumn{1}{c}{\Enwiki} & \multicolumn{1}{c}{\Influenza} & \multicolumn{1}{c}{\Swissprot} \\
\noalign{\smallskip}
\hline
\noalign{\smallskip}
\BruteD & 59.419 \textmu{}s & \textbf{124.286} \textmu{}s & \textbf{714.481} \textmu{}s & \textbf{4557.310} \textmu{}s & \textbf{9.392} \textmu{}s \\
 & 0.000 \busw & \textbf{0.000} \busw & \textbf{0.000} \busw & \textbf{0.000} \busw & \textbf{0.000} \busw \\
\noalign{\smallskip}
\PDLRP & \textbf{43.356} \textmu{}s & 217.804 \textmu{}s & 1107.470 \textmu{}s & 6221.610 \textmu{}s & 24.848 \textmu{}s \\
 & \textbf{0.000} \busw & 0.000 \busw & 0.000 \busw & 0.000 \busw & 0.000 \busw \\
\noalign{\smallskip}
\hline
\noalign{\smallskip}
\Sada & \textbf{0.218} \textmu{}s & \textbf{0.213} \textmu{}s & \textbf{0.250} \textmu{}s & 0.624 \textmu{}s & \textbf{0.246} \textmu{}s \\
 & \textbf{2.094} \busw & \textbf{2.094} \busw & \textbf{2.094} \busw & 2.093 \busw & \textbf{2.091} \busw \\
\noalign{\smallskip}
\SadaPG & 2.030 \textmu{}s & 1.442 \textmu{}s & 1.608 \textmu{}s & 1.291 \textmu{}s & \multicolumn{1}{c}{--} \\
 & 1.307 \busw & 2.469 \busw & 2.694 \busw & 2.466 \busw & \multicolumn{1}{c}{--} \\
\noalign{\smallskip}
\SadaPR & \textbf{0.852} \textmu{}s & 0.882 \textmu{}s & 1.572 \textmu{}s & 1.356 \textmu{}s & \multicolumn{1}{c}{--} \\
 & \textbf{0.146} \busw & 2.455 \busw & 2.748 \busw & 2.466 \busw & \multicolumn{1}{c}{--} \\
\noalign{\smallskip}
\hline
\noalign{\smallskip}
\SadaR & 1.105 \textmu{}s & 0.506 \textmu{}s & 1.013 \textmu{}s & 0.581 \textmu{}s & \textbf{0.779} \textmu{}s \\
 & 5.885 \busw & 0.125 \busw & 1.223 \busw & 0.007 \busw & \textbf{0.076} \busw \\
\noalign{\smallskip}
\SadaRG & 2.268 \textmu{}s & 1.535 \textmu{}s & 2.001 \textmu{}s & 1.046 \textmu{}s & \multicolumn{1}{c}{--} \\
 & 1.297 \busw & 0.070 \busw & 1.088 \busw & 0.007 \busw & \multicolumn{1}{c}{--} \\
\noalign{\smallskip}
\SadaRR & \textbf{1.088} \textmu{}s & 0.974 \textmu{}s & 1.960 \textmu{}s & 1.108 \textmu{}s & \multicolumn{1}{c}{--} \\
 & \textbf{0.136} \busw & 0.056 \busw & 1.142 \busw & 0.007 \busw & \multicolumn{1}{c}{--} \\
\noalign{\smallskip}
\hline
\noalign{\smallskip}
\SadaG & 23.750 \textmu{}s & \textbf{21.643} \textmu{}s & 18.542 \textmu{}s & 33.502 \textmu{}s & \textbf{25.236} \textmu{}s \\
 & 0.086 \busw & \textbf{0.024} \busw & 0.439 \busw & 0.005 \busw & \textbf{0.034} \busw \\
\noalign{\smallskip}
\hline
\noalign{\smallskip}
\sadaR & 0.742 \textmu{}s & 0.396 \textmu{}s & 0.688 \textmu{}s & 0.584 \textmu{}s & 0.538 \textmu{}s \\
 & 5.991 \busw & 0.222 \busw & 1.180 \busw & 0.006 \busw & 0.082 \busw \\
\noalign{\smallskip}
\sadaRS & 0.897 \textmu{}s & \textbf{0.492} \textmu{}s & \textbf{0.923} \textmu{}s & \textbf{0.767} \textmu{}s & 0.545 \textmu{}s \\
 & 1.042 \busw & \textbf{0.059} \busw & \textbf{0.424} \busw & \textbf{0.005} \busw & 0.082 \busw \\
\noalign{\smallskip}
\hline
\noalign{\smallskip}
\sadaD & 1.019 \textmu{}s & 0.521 \textmu{}s & 1.119 \textmu{}s & 0.856 \textmu{}s & \textbf{0.792} \textmu{}s \\
 & 3.717 \busw & 0.088 \busw & 0.942 \busw & 0.006 \busw & \textbf{0.062} \busw \\
\noalign{\smallskip}
\sadaDS & 1.205 \textmu{}s & \textbf{0.641} \textmu{}s & \textbf{1.316} \textmu{}s & 1.005 \textmu{}s & 0.799 \textmu{}s \\
 & 0.989 \busw & \textbf{0.046} \busw & \textbf{0.374} \busw & 0.005 \busw & 0.062 \busw \\
\noalign{\smallskip}
\hline
\noalign{\smallskip}
\sadaSS & 0.604 \textmu{}s & \textbf{0.269} \textmu{}s & \textbf{0.525} \textmu{}s & \textbf{0.439} \textmu{}s & \textbf{0.396} \textmu{}s \\
 & 5.729 \busw & \textbf{0.209} \busw & \textbf{1.079} \busw & \textbf{0.006} \busw & \textbf{0.078} \busw \\
\noalign{\smallskip}
\sadaS & 0.735 \textmu{}s & \textbf{0.380} \textmu{}s & \textbf{0.755} \textmu{}s & 0.624 \textmu{}s & 0.399 \textmu{}s \\
 & 3.432 \busw & \textbf{0.142} \busw & \textbf{0.823} \busw & 0.006 \busw & 0.078 \busw \\
\noalign{\smallskip}
\hline
\noalign{\smallskip}
\wt & 4.399 \textmu{}s & 4.482 \textmu{}s & 6.033 \textmu{}s & 7.252 \textmu{}s & 3.414 \textmu{}s \\
 & 18.454 \busw & 0.484 \busw & 4.575 \busw & 0.525 \busw & 0.992 \busw \\
\noalign{\smallskip}
\hline
\end{tabular}
\end{table}

\section{Index Construction}\label{appendix:construction}

Our construction algorithms prioritize flexibility over performance. For example, the construction of the \tfidf{} index (Section~\ref{sec:tfidf}) proceeds as follows:
\begin{enumerate}

\item Build RLCSA for the collection.

\item Extract the LCP array and the document array from the RLCSA, traverse the suffix tree by using the LCP array, and build \PDL{} with uncompressed document sets.

\item Compress the document sets using a Re-Pair compressor.

\item Build the \sadaSS{} structure using a similar algorithm as for \PDL{} construction.

\end{enumerate}
See Table~\ref{table:construction} for the time and space requirements of building the index for the \Wiki{} collection.

\begin{table}[t!]\centering
\caption{Building the \tfidf{} index for the \Wiki{} collection. Construction time in minutes and peak memory usage in gigabytes for RLCSA construction, \PDL{} construction, compressing the document sets using Re-Pair, \sadaSS{} construction, and the entire construction.}\label{table:construction}

\begin{tabular}{cccccc}
\hline
\noalign{\smallskip}
       &    RLCSA &     \PDL & Re-Pair &  \sadaSS & Total \\
\noalign{\smallskip}
\hline
\noalign{\smallskip}
Time   & 10.5 min & 39.2 min & 123 min & 74.7 min & 248 min \\
\noalign{\smallskip}
Memory &  19.6 GB &   111 GB &  202 GB &  92.8 GB & 202 GB\\
\noalign{\smallskip}
\hline
\end{tabular}
\end{table}

Scaling the index up for larger collections requires faster and more space-efficient construction algorithms for its components. There are some obvious improvements:
\begin{itemize}

\item RLCSA construction can be done in less memory by building the index in multiple parts and merging the partial indexes \citep{Siren2009}. With 100~parts, the indexing of a repetitive collection proceeds at about 1~MB/s using 2\nobreakdash--3~bits per symbol \citep{Siren2012}. Newer suffix array construction algorithms achieve even better time/space trade-offs \citep{Kaerkkaeinen2015a}.

\item We can use a compressed suffix tree for \PDL{} construction. The SDSL library \citep{Gog2014b} provides fast scalable implementations that require around 2~bytes per symbol.

\item We can write the uncompressed document sets to disk as soon as the traversal returns to the parent node.

\item We can build the $H$ array for \sadaSS{} by keeping track of the lowest common ancestor of the previous occurrence of each document identifier and the current node. If node $v$ is the lowest common ancestor of consecutive occurrences of a document identifier, we increment the corresponding cell of the $H$ array. Storing the array requires about a byte per symbol.

\end{itemize}

The main bottleneck in the construction is Re-Pair compression. Our compressor requires 24~bytes of memory for each integer in the document sets, and the number of integers (8.9~billion) is several times larger than the number of symbols in the collection (1.5~billion). It might be possible to improve compression performance by using a specialized compressor. If interval $\DA[\ell..r]$ corresponds to suffix tree node $u$ and the collection is repetitive, it is likely that the interval $\DA[\ell'..r']$ corresponding to the node reached by taking the suffix link from $u$ is very similar to $\DA[\ell..r]$.

\end{document}